\documentclass{llncs}
\usepackage[utf8]{inputenc}
\usepackage{subfigure}
\usepackage{amssymb}
\usepackage{amsmath}
\usepackage{xspace}
\usepackage{paralist}
\usepackage{graphicx}
\usepackage{color}
\usepackage[unicode]{hyperref}
\usepackage{todonotes}
\usepackage{amssymb}

\newcommand{\rephrase}[3]{\noindent\textbf{#1 #2}.~\emph{#3}}
\newcommand{\remove}[1]{}

\newenvironment{sketch}{\noindent{\em Proof sketch.}}{\hspace*{\fill}$\square$\vspace{2mm}}

\let\doendproof\endproof
\renewcommand\endproof{~\hfill\qed\doendproof}

\DeclareMathOperator{\pert}{pert}
\DeclareMathOperator{\skel}{skel}
\DeclareMathOperator{\twin}{twin}
\newcommand{\eps}{\ensuremath{\varepsilon}}

\definecolor{blue}{rgb}{0.274,0.392,0.666}
\definecolor{red}{rgb}{0.627,0.117,0.156}
\definecolor{Red}{rgb}{1,0,0}
\definecolor{green}{rgb}{0,0.588,0.509}
\definecolor{gray}{rgb}{0.5,0.5,0.5}

\newcommand{\ourproblem}[1]{{\sc OrthoSEFE-$#1$}\xspace}
\newcommand{\ourdrawing}{OrthoSEFE\xspace}
\newcommand{\ourdrawings}{OrthoSEFEs\xspace}
\newcommand{\ourinstance}[1]{$\langle G_1^{#1}, G_2^{#1}\rangle$\xspace}
\newcommand{\ourinstancek}[1]{$\langle G_1^{#1}, \dots, G_k^{#1}\rangle$\xspace}
\newcommand{\oursolution}[1]{$\langle \Gamma_1^{#1}, \Gamma_2^{#1}\rangle$\xspace}
\newcommand{\oursolutionk}[1]{$\langle \Gamma_1^{#1}, \Gamma_2^{#1}, \dots, \Gamma_k^{#1}\rangle$\xspace}
\newcommand{\naesat}{\textsc{Nae3Sat}\xspace}
\newcommand{\pnaesat}{\textsc{Planar Nae3Sat}\xspace}
\newcommand{\true}{{\sf true}\xspace}
\newcommand{\false}{{\sf false}\xspace}

\newcommand{\nae}{\textsc{Nae}}
\begin{document}

\title{Simultaneous Orthogonal Planarity\thanks{This research was initiated at the
Bertinoro Workshop on Graph Drawing 2016. 
Research was partially supported by DFG grant Ka812/17-1, by MIUR project AMANDA, prot. 2012C4E3KT\_001, by the grant no. 14-14179S of the Czech Science Foundation GACR, and by DFG grant WA 654/21-1.
}
}

\author{
  Patrizio~Angelini\inst{1}
  \and
  Steven~Chaplick\inst{2}
  \and 
  Sabine~Cornelsen\inst{3}
  \and
  Giordano~{\mbox{Da Lozzo}}\inst{4}
  \and
  Giuseppe~Di Battista\inst{4}
  \and
  Peter~Eades\inst{5}
  \and
  Philipp~Kindermann\inst{6}
  \and
  Jan~Kratochv\'il\inst{7}
  \and
  Fabian~Lipp\inst{2}
  \and
  Ignaz~Rutter\inst{8}
}

\institute{
  Universit\"at T\"ubingen, Germany -- \email{angelini@informatik.uni-tuebingen.de}
  \and
  %
  Universit\"at W\"urzburg, Germany -- \email{\{steven.chaplick,fabian.lipp\}@uni-wuerzburg.de}
  \and
  %
  Konstanz University, Germany -- \email{sabine.cornelsen@uni-konstanz.de}
  \and
  %
  Roma Tre University, Italy -- \email{\{dalozzo,gdb\}@dia.uniroma3.it}
  \and
  %
  The University of Sydney, Australia -- \email{peter@it.usyd.edu.au}
  \and
  %
  FernUniversit\"at in Hagen, Germany -- \email{philipp.kindermann@fernuni-hagen.de}
  \and
  %
  Charles University, Czech Republic -- \email{honza@kam.mff.cuni.cz}
  \and
  %
  Karlsruhe Institute of Technology, Germany -- \email{rutter@kit.edu}
}

\maketitle

\begin{abstract}
  We introduce and study the \ourproblem{k} problem: 
  Given~$k$ planar graphs each with maximum degree 4 and the same vertex set,
  do they admit an OrthoSEFE, that is,
  is there an assignment of the vertices to
  grid points and of the edges to paths on the grid such that the same
  edges in distinct graphs are assigned the same path and such that
  the assignment induces a planar orthogonal drawing of each of the~$k$
  graphs?

  We show that the problem is NP-complete for $k \geq 3$ even
  if the shared graph is a Hamiltonian cycle and has sunflower intersection and for $k \geq 2$ even
  if the shared graph consists of a cycle and of isolated vertices. 
  Whereas the problem is polynomial-time solvable for $k=2$ when the
  union graph has maximum degree five and the shared graph is
  biconnected.  Further, when the shared graph is biconnected and has
  sunflower intersection, we show that every positive instance has an OrthoSEFE
  with at most three bends per edge.
\end{abstract}

\section{Introduction}
\label{se:introduction}

The input of a simultaneous embedding problem consists of several
graphs $G_1=(V,E_1),\dots,G_k=(V,E_k)$ on the same vertex set.  For a
fixed drawing style $\mathcal S$, the simultaneous embedding problem
asks whether there exist drawings $\Gamma_1,\dots,\Gamma_k$ of
$G_1,\dots,G_k$, respectively, in drawing style $\mathcal S$ such that for any
$i$ and $j$ the restrictions of $\Gamma_i$ and $\Gamma_j$ to $G_i \cap
G_j = (V,E_i \cap E_j)$ coincide.

The problem has been most widely studied in the setting of topological
planar drawings, where vertices are represented as points and edges
are represented as pairwise interior-disjoint Jordan arcs between
their endpoints.  This problem is called {\sc Simultaneous
  Embedding with Fixed Edges} or {\sc SEFE-$k$} for short, where $k$
is the number of input graphs.  It is known that {\sc SEFE-$k$} is
NP-complete for $k \geq 3$, even in the restricted case
of {\em sunflower instances}~\cite{s-ttphtpv-13}, where every pair of graphs shares the same
set of edges, and even if such a set induces a star~\cite{adn-aspbep-15}.  
On the other hand, the complexity for $k=2$ is
still open. Recently, efficient algorithms for
restricted instances have been presented, namely when
\begin{inparaenum}[(i)]
\item the shared graph $G_\cap = G_1 \cap G_2$ is
  biconnected~\cite{hjl-tspcg2c-10,adfpr-tsetgibgt-11} or a
  star-graph~\cite{adfpr-tsetgibgt-11},
\item $G_\cap$ is a collection of disjoint cycles~\cite{br-drpse-15},
\item every connected component of $G_\cap$ is either subcubic or 
  biconnected~\cite{s-ttphtpv-13,bkr-seeor-15},
\item $G_1$ and $G_2$ are biconnected and $G_\cap$ is 
  connected~\cite{br-spqoacep-13}, and  
\item $G_\cap$ is connected and the input graphs have maximum
  degree~5~\cite{br-spqoacep-13}; see the survey by Bläsius et
  al.~\cite{bkr-sepg-12} for an overview.
\end{inparaenum}

For planar straight-line drawings, the simultaneous embedding problem
is called {\sc Simultaneous Geometric Embedding} and it is known to be
NP-hard 
even for two graphs~\cite{egjpss-sgge-07}.
Besides simultaneous intersection
representation for, e.g., interval
graphs~\cite{simultaneous_interval_graphs,br-spqoacep-13} and
permutation and chordal graphs~\cite{jl-srpcc-12}, it is only recently
that the simultaneous embedding paradigm has been applied to other
fundamental planarity-related drawing styles, namely simultaneous level
planar drawings~\cite{addfpr-blp-15} and RAC
drawings~\cite{abks-gracsdg-13,bdkw-sdpgrac-16}.  

We continue this line of research by studying
simultaneous embeddings in the planar {\em orthogonal} drawing style, where
vertices are assigned to grid points and edges to paths on the grid connecting their endpoints~\cite{t-eggmdb-87}.  In
accordance with the existing naming scheme, we define \ourproblem{k} to be the problem of testing whether $k$ input graphs $\langle G_1,\dots,G_k \rangle$
admit a simultaneous planar orthogonal drawing.  If
such a drawing exists, we call it an OrthoSEFE of $\langle
G_1,\dots, G_k \rangle$.  Note that it is a necessary condition that
each $G_i$ has maximum degree~4 in order to obtain planar orthogonal
drawings. Hence, in the remainder of the paper we assume that all instances have this property.
For instances with this property, at least when the shared graph is
connected, the problem {\sc SEFE}-2 can be solved efficiently~\cite{br-spqoacep-13}.  However, there are
instances of \ourproblem{2} that admit a SEFE but not an
OrthoSEFE; see Fig.~\ref{fig:negative-instance}.

Unless mentioned otherwise, all instances of \ourproblem{k} and {\sc SEFE}-$k$ we consider are sunflower. Notice that instances with $k=2$ are always sunflower. Let $\langle G_1=(V,E_1),G_2=(V,E_2) \rangle$ be an instance of \ourproblem{2}. We define the \emph{shared graph} (resp. the \emph{union graph}) to be the graph 
$G_\cap = (V, E_1 \cap E_2)$ (resp. $G_\cup = (V, E_1 \cup E_2)$) with the same vertex 
set as~$G_1$ and~$G_2$, whose edge set is the 
intersection (resp. the union) of the ones of~$G_1$ and~$G_2$. Also, we call 
the edges in $E_1 \cap E_2$ the \emph{shared edges} and we call the edges 
in $E_1 \setminus E_2$ and in $E_2 \setminus E_1$ the \emph{exclusive edges}. 
The definitions of \emph{shared graph}, \emph{shared edges}, and
\emph{exclusive edges} naturally extend to sunflower instances for any value of $k$.

\begin{figure}[tb]
  \centering
  \subfigure[]{
  \includegraphics[]{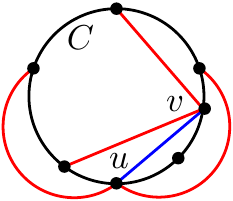}\label{fig:negative-instance}
  }
  \hfil
  \subfigure[]{
  \includegraphics[]{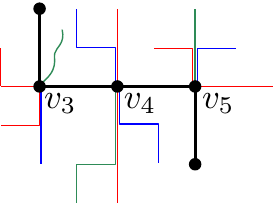}\label{fig:corners}
  }
  \caption{(a) A negative instance of \ourproblem{2}. Shared edges are black, while exclusive edges are red and blue. The red edges require $270^\circ$ angles on different sides of $C$.
  Thus, the blue edge $(u,v)$ cannot be drawn. 
  Note that the given drawing is a {\sc SEFE}-2.
  (b) Examples of side assignments for the exclusive edges incident to degree-$2$ vertices of $G_\cap$: orthogonality constraints are satisfied at $v_4$ and $v_5$, while they are violated at~$v_3$.
  }
  
\end{figure}

One main issue is to decide how degree-$2$ vertices of the shared
graph are represented.  Note that, in planar topological drawings,
degree-$2$ vertices do not require any decisions as there exists only a
single cyclic order of their incident edges.  In the case of
orthogonal drawings there are, however, two choices for a degree-$2$
vertex: It can either be drawn {\em straight}, i.e., it is incident to two
angles of $180^\circ$, or \emph{bent}, i.e., it is incident to
one angle of $90^\circ$ and to one angle of $270^\circ$.  If $v$ is a
degree-$2$ vertex of the shared graph with neighbors $u$ and $w$, and
two exclusive edges $e,e'$, say of $G_1$, are incident to $v$ and are
embedded on the same side of the path $uvw$, then $v$ must be bent,
which in turn implies that also every exclusive edge of $G_2$ incident
to $v$ has to be embedded on the same side of $uvw$ as $e$ and $e'$.
In this way, the two input graphs of \ourproblem{2} interact via the
degree-$2$ vertices.  It is the difficulty of
controlling this interaction that marks the main difference between {\sc SEFE}-$k$ and
\ourproblem{k}.  To study this interaction in isolation, we focus on
instances of \ourproblem{2} where the shared graph is a cycle for most
of the paper.  Note that such instances are trivial yes-instances of
{\sc SEFE}-$k$ (provided the input graphs are all planar).

{\bf Contributions and Outline.} In Section~\ref{se:preliminaries} we provide our notation and
we show that the existence of an \ourdrawing of an
instance of \ourproblem{k} can be described as a combinatorial
embedding problem.
In Section~\ref{se:cycle-hardness}, we show that \ourproblem{3} is
NP-complete even if the shared graph is a cycle, and that
\ourproblem{2} is NP-complete even if the shared graph consists of a
cycle plus some isolated vertices.  This contrasts the situation of
{\sc SEFE}-$k$ where these cases are polynomially
solvable~\cite{adfpr-tsetgibgt-11,bkr-seeor-13,hjl-tspcg2c-10,s-ttphtpv-13}.
In Section~\ref{se:cycle-algorithms}, we show that \ourproblem{2} is
efficiently solvable if the shared graph is a cycle and the union
graph has maximum degree~$5$.  Finally, in
Section~\ref{se:biconnected}, we extend this result to the case where
the shared graph is biconnected (and the union graph still has maximum
degree~$5$).  
%
%
Moreover, we show
that any positive instance of \ourproblem{k} whose shared graph is
biconnected admits an \ourdrawing with at most three bends per edge.
We close with some concluding remarks and open questions in
Section~\ref{se:conclusions}.

Full proofs can be found in the Appendix.

\section{Preliminaries} \label{se:preliminaries}

\remove{Let $G=(V,E)$ be a planar graph.  An {\em orthogonal drawing} $\Gamma$
of $G$ maps each vertex in $V$ to a grid point and each edge to a path
in the grid in such a way that the resulting drawing is planar.  Since
each edge in $\Gamma$ is an alternating sequence of vertical and
horizontal segments, every vertex has four available {\em ports}, each
corresponding to one of the possible segments incident to the vertex.
Clearly, $G$ admits an orthogonal drawing if and only if no vertex has
degree larger than $4$, i.e., $G$ is $4$-planar. \todo{Maybe add a reference? We could use, e.g.,  \cite{biedl/kant:98}. But isn't that fact folklore?}
} 

We will extensively make use of the {\sc Not-All-Equal
$3$-Sat} (\naesat) problem~\cite[p.187]{Pap07}. An instance of \naesat
consists of a $3$-CNF formula $\phi$ with variables $x_1,\ldots,x_n$
and clauses $c_1,\ldots,c_m$. The task is to find a
\emph{\textsc{Nae} truth assignment}, i.e., a truth assignment such
that each clause contains both a true and a false literal. \naesat is
known to be NP-complete~\cite{Schaefer-tcsp-78}. The
\emph{variable--clause graph} is the bipartite graph whose vertices are
the variables and the clauses, and whose edges represent the
membership of a variable in a clause. The problem
\pnaesat is the restriction of \naesat to instances whose
variable--clause graph is planar. \pnaesat can be solved 
efficiently~\cite{moret-pnp-88,shih_etal:90}.

\paragraph{\bf Embedding Constraints.}\label{se:embedding-constraints}

Let \ourinstancek{} be an \ourproblem{k} instance. A {\em SEFE} is a collection of embeddings $\mathcal E_i$ for the $G_i$ such that their restrictions on $G_\cap$ are the same.  Note that in the literature, a SEFE is often defined as a collection of drawings rather than a collection of embeddings.  However, the two definitions are equivalent~\cite{js-igsefe-09}. 
For a SEFE to be realizable as an \ourdrawing it needs to satisfy two additional conditions.
First, let~$v$ be a vertex of degree~2 in~$G_\cap$ with neighbors $u$ and $w$. If in any embedding $\mathcal E_i$ there exist two
exclusive edges incident to $v$ that are embedded on the same side of the 
path $uvw$, then any exclusive edge incident to~$v$ in any of 
the $\mathcal E_j \neq \mathcal E_i$ must be embedded on the same side of the path $uvw$.
Second, let~$v$ be a vertex of degree~3 in~$G_\cap$. All exclusive edges incident 
to $v$ must appear between the same two edges of $G_\cap$ around $v$. We call 
these the {\em orthogonality constraints}. See Fig.~\ref{fig:corners}.

\begin{theorem}\label{th:general-characterization}
An instance \ourinstancek{} of \ourproblem{k} has an \ourdrawing if and only if it admits a SEFE satisfying the orthogonality constraints.
\end{theorem}

For the case in which the shared graph is a cycle~$C$, we give a simpler version of the constraints in Theorem~\ref{th:general-characterization}, which will prove useful in the remainder of the paper.
By the Jordan curve Theorem, a planar drawing of cycle $C$ divides the plane into a bounded and
an unbounded region~--~the \emph{inside} and the \emph{outside} of
$C$, which we call the \emph{sides}
of $C$. Now the problem is to assign the exclusive edges to either of the
two sides of $C$ so that the following two conditions are fulfilled.

\smallskip
\noindent
{\em Planarity Constraints.} Two exclusive edges of the same
  graph must be drawn on different sides of $C$ if their endvertices
  alternate along $C$.
  \\
\noindent
{\em Orthogonality Constraints.} Let $v \in V$ be a vertex that is
  adjacent to two exclusive edges $e_i$ and $e'_i$ of the same graph $G_i$,
  $i \in \{1,\dots,k\}$. 
  If $e_i$ and $e'_i$ are on the same side of $C$, then
  all exclusive edges incident to $v$ of all graphs $G_1,\dots,G_k$
  must be on the same side as $e_i$ and $e'_i$. 
\smallskip

Note that
this is a reformulation of the general orthogonality
constraints. Further, the orthogonality constraints also imply that if
$e_i$ and $e_i'$ are on different sides of $C$, then for each graph
$G_j$ that contains two exclusive edges $e_j$ and $e_j'$ incident to
$v$, with $j \in \{1,\dots, k\}$,  $e_j$ and $e_j'$ must
be on different sides of $C$.


The next theorem follows from Theorem~\ref{th:general-characterization} and from the following two observations. First, for a sunflower instance \ourinstancek{} whose shared graph is a cycle, any collection of embeddings is a SEFE~\cite{js-igsefe-09}. Second, the planarity constraints are necessary and sufficient for the existence of an embedding of $G_i$~\cite{ap-oigs-61}.

\begin{theorem}\label{th:characterization}
  An instance 
  of \ourproblem{k} whose shared graph
  is a cycle $C$ 
  has an
  \ourdrawing if and only if there exists an assignment of the
  exclusive edges 
  to the two sides of $C$ satisfying
   the planarity and orthogonality constraints.
\end{theorem}

\section{Hardness Results}\label{se:cycle-hardness}

We show that \ourproblem{k} is NP-complete for $k \geq
3$ for instances with sunflower intersection even if the shared graph
is a cycle, and for $k=2$ even if the shared graph consists of a cycle and isolated vertices.

\begin{theorem}\label{th:hardness-hamiltonian-cycle-three-colors}
  \ourproblem{k} with $k \geq 3$ is NP-complete, even for instances with
  sunflower intersection in which (i) the shared
  graph is a cycle and (ii) $k-1$ of the input graphs are outerplanar and have maximum degree $3$.
\end{theorem}

\begin{sketch}
  The membership in NP directly follows from Theorem~\ref{th:characterization}.
  To prove the NP-hardness, we show a polynomial-time reduction from the
  NP-complete problem {\sc Positive Exactly-Three}
  \naesat~\cite{moret-toc-98}, which is the variant of \naesat in
  which each clause consists of exactly three unnegated literals.

\begin{figure}[tb]
	\centering
  \subfigure[]{\includegraphics[height=0.3\columnwidth,page=3]{img/clause-gadget.pdf}\label{fig:hardness-3-graphs-a}}
  \hfil
	\subfigure[]{\includegraphics[height=0.3\columnwidth]{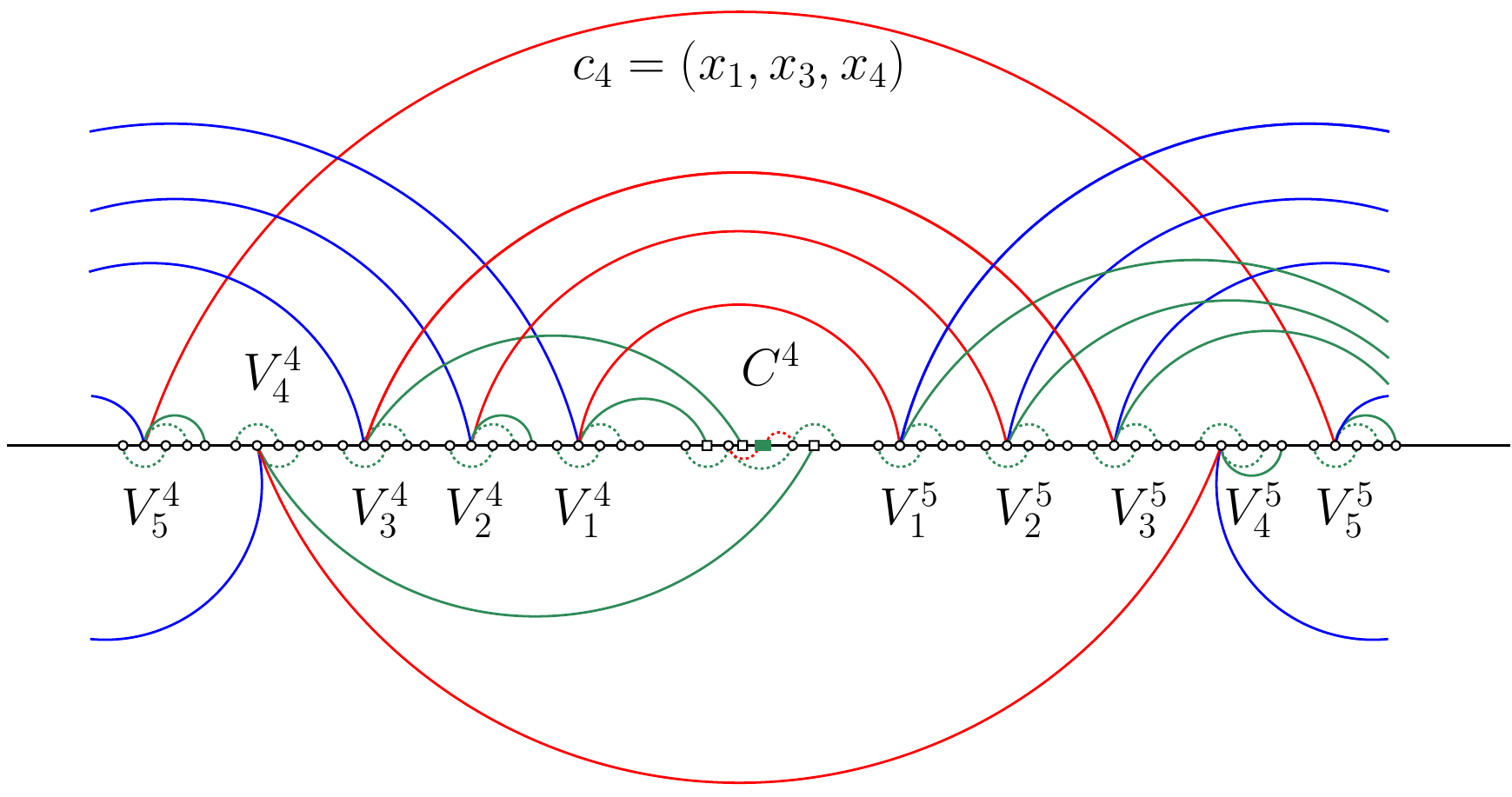}\label{fig:hardness-3-graphs-b}}
	\caption{(a) A clause gadget
      $C_j$ (top) and a variable-clause gadget $V_i^j$ (bottom); solid edges belong to the gadgets, dotted edges are optional, and dashed edges are transmission edges.  
      (b) Illustration of instance $\langle G_1,G_2,G_3\rangle$, focused on a clause $c_4$. Black edges
      belong to the shared graph $G_\cap$. The red, blue, and green edges are the exclusive edges of $G_1$,
      $G_2$, and $G_3$, respectively.}
      \label{fig:hardness-3-graphs}
\end{figure}

Let $x_1,x_2,\dots,x_n$ be the variables and let $c_1,c_2,\dots,c_m$
be the clauses of a $3$-CNF formula $\phi$ of {\sc Positive
  Exactly-Three} \naesat. We show how to construct an equivalent
instance $\langle G_1,G_2,G_3\rangle$ of
\ourproblem{3} such that $G_1$ and $G_2$ are outerplanar graphs of
maximum degree 3. We refer to the exclusive edges in $G_1$, $G_2$, and
$G_3$ as red, blue, and green, respectively; refer to
Fig.~\ref{fig:hardness-3-graphs}.


For each clause $c_j$, $j=1,\dots,m$, we create a \emph{clause gadget}
$C^j$ as in Fig.~\ref{fig:hardness-3-graphs-a} (top).  
For each variable $x_i$, $i=1,\dots,n$, and each clause~$c_j$, $j=1,\dots,m$, we
create a \emph{variable-clause gadget} $V_i^j$ as in
Fig.~\ref{fig:hardness-3-graphs-a} (bottom).
Observe that the (dotted) green edge $\{w_i^j,r_i^j\}$ in a variable-clause
gadget is only part of $V_i^j$ if $x_i$ does not occur in $c_j$.
Otherwise, there is a green edge $\{w_i^j,y_x^j\}$ connecting
$w_i^j$ to one of the three vertices $y_a^j$, $y_b^j$, or $y_c^j$ (dashed stubs) in the
clause gadget. Observe that these
three \emph{variable-clause edges} per clause can be realized in such a way
that there exist no planarity constraints between pairs of them. In
Fig.~\ref{fig:hardness-3-graphs-b}, the variable-clause gadgets $V_1^4$,
$V_3^4$, $V_4^4$ are incident to variable-clause edges, while $V_2^4$ and
$V_5^4$ contain edges $\{w_2^4,r_2^4\}$ and $\{w_5^4,r_5^4\}$,  respectively.

The gadgets are ordered as indicated in Fig.~\ref{fig:hardness-3-graphs-b}. 
The variable-clause gadgets~$V_i^j$, with $i=1,\dots,n$, always precede the 
clause gadget $V^j$, for any $j=1,\dots,m$.  Further, if $j$ is odd, then the 
gadgets $V_1^j,\dots,V_n^j$ appear in this order, otherwise they appear in 
reversed order $V_n^j,\dots,V_1^j$.
Finally, $V_i^j$ and $V_i^{j+1}$, for $i=1,\dots,n$ and $j=1,\dots,m-1$, are connected by an edge $\{w_i^j,w_i^{j+1}\}$, which is blue if $j$ is
odd and red if $j$ is even. We call these edges {\em transmission edges}.

Assume $\langle G_1, G_2, G_3 \rangle$ admits an \ourdrawing.
Planarity constraints and orthogonality constraints guarantee three
properties: 
(i) If the edge $\{u_i^j,v_i^j\}$ is inside~$C$, then so 
is $\{u_i^{j+1},v_i^{j+1}\}$, $i=1,\dots,n$, $j=1,\dots,m-1$. This is due to the 
fact that, by the planarity constraints, the two green edges incident to~$w_i^j$
lie on the same side of~$C$ and hence, by the orthogonality constraints, the two 
transmission edges incident~to $w_i^j$ also lie in this side. We
call $\{u_i^1,v_i^1\}$ the \emph{truth edge} of variable~$x_i$.
(ii) Not all the three green edges $a=\{\alpha^j,\beta^j\}$, 
$b=\{\beta^j,\gamma^j\}$, and $c=\{\gamma^j,\delta^j\}$ lie on the same side 
of~$C$. Namely, the two red edges of the clause gadget~$C^j$ must lie on 
opposite sides~of $C$ because of the interplay between the planarity and the 
orthogonality constraints in the subgraph of~$C^j$ induced by the vertices 
between~$\beta^j$ and~$\gamma^j$. Hence, if edges~$a$, $b$, and~$c$ lie in the 
same side of~$C$, then the orthogonality constraints at either~$\beta^j$ 
or~$\gamma^j$ are not satisfied.
(iii) For each clause $c_j=(x_a,x_b,x_c)$, edge $a=\{\alpha^j,\beta^j\}$ lies in
the same side of~$C$ as the truth edge of $x_a$. This is due to the planarity
constraints between each of these two edges and the variable-clause edge
$\{w_a^j,y_a^j\}$. Analogously, edge $b$ (edge $c$) lies on the same side as the
truth edge of $x_b$ (of $x_c$).
Hence, setting $x_i = \true$ ($x_i = \false$) if the truth edge of $x_i$ is inside $C$ (outside $C$) yields a \naesat truth assignment that satisfies $\phi$.

The proof for the other direction is based on the fact that assigning the truth edges to either of the two sides of $C$ according to the \naesat assignment of $\phi$
also implies a unique side assignment for the remaining exclusive edges that
satisfies all the orthogonality and the planarity constraints. 

\remove{Assume on the other hand that $\phi$ has a \naesat truth
assignment. Assign the truth edge of a variable $x_i$ to the inside (outside)
of $C$ if $x_i=\true$ ($x_i=\false$). By the properties
described above, assigning the truth edges to either of the two sides of $C$
also implies a unique side assignment for the remaining exclusive edges that
satisfies all the orthogonality and the planarity constraints. In particular, for
each $j=1,\dots,m$, the drawing of the clause gadget $C^j$ respects all the
constraints, since edges $\{\alpha^j,\beta^j\}$, $\{\beta^j,\gamma^j\}$, and
$\{\gamma^j,\delta^j\}$ do not all lie on the same side of $C$, due to the fact that
$\phi$ has a \naesat assignment.}

It is easy to see that $G_1$ and $G_2$ are outerplanar graphs with maximum degree~$3$, and that the reduction can be extended to any $k>3$.
\end{sketch}

In the following we describe how to modify the construction in Theorem~\ref{th:hardness-hamiltonian-cycle-three-colors} to show hardness of \ourproblem{2}. We keep only the
  edges of $G_1$ and $G_3$. Variable-clause gadgets and clause gadgets remain the same, as they are composed only of edges belonging to these two graphs. We replace each transmission edge in $G_2$ by a {\em transmission path} composed of alternating green and red edges, starting and ending with a red edge. 
  This transformation allows these paths to traverse the transmission edges of $G_1$ and the variable-clause edges of $G_3$ without introducing crossings between edges of the same color. 
  It is easy to see that the properties described in the proof of Theorem~\ref{th:hardness-hamiltonian-cycle-three-colors} on the assignments of the exclusive edges to the two sides of $C$ also hold in the constructed instance, where transmission paths take the role of the transmission edges.

\begin{theorem}\label{th:hardness-cycle-two-colors}
\ourproblem{2} is NP-complete, even for instances \ourinstance{} in which the shared graph consists of a cycle and a set of isolated vertices.
\end{theorem}

\section{Shared Graph is a Cycle}
\label{se:cycle-algorithms}

In this section we give a polynomial-time algorithm for instances of \ourproblem{2} 
whose shared graph is a cycle and whose union graph has maximum degree $5$ (Theorem~\ref{th:algorithm-one-degree-three}).
In order to obtain this result, we present an efficient algorithm for more restricted instances (Lemma~\ref{le:algorithm-degree-three-outerplanar}) and give a series of transformations (Lemma~\ref{le:reduction-degree-three-2-degree-three-outerplanar}--\ref{le:reduction-degree-five-2-degree-three}) to reduce any instance with the above properties to one that can be solved by the algorithm in Lemma~\ref{le:algorithm-degree-three-outerplanar}.

\begin{lemma}\label{le:algorithm-degree-three-outerplanar}
  \ourproblem{2} is in P
   for instances \ourinstance{} 
  such that the shared graph $C$ is a cycle and $G_1$ is an
  outerplanar graph with maximum degree~$3$.
\end{lemma}

\begin{proof}
\begin{figure}[tb]
  \centering
  \subfigure[]{
  \includegraphics[page=5]{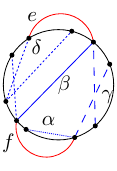}\label{fig:cycle-algorithm-a}}
  \hfil
    \subfigure[]{\includegraphics[page=2]{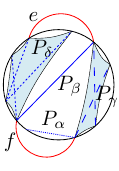}\label{fig:cycle-algorithm-b}}
  \hfil
  \subfigure[]{\includegraphics[page=3]{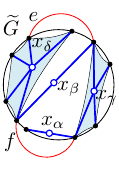}\label{fig:cycle-algorithm-c}}
    \hfil
  \subfigure[]{\includegraphics[page=4]{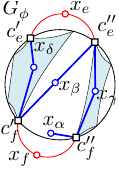}\label{fig:cycle-algorithm-d}}
  \caption{(a)  Instance \ourinstance{} satisfying the properties of Lemma~\ref{le:algorithm-degree-three-outerplanar}, where the edges in $E_2$ belonging to the components $\alpha$, $\beta$, $\gamma$, and $\delta$ of $H$ have different line styles. 
  (b) Polygons for the components of $H$. 
  (c) Graph $\widetilde G$.
  (d) Variable--clause graph~$G_\phi$.}
  \label{fig:cycle-algorithm}
\end{figure}

The algorithm is based on a reduction to \pnaesat, which is in P~\cite{moret-pnp-88,shih_etal:90}.
First note that since $G_1$ is outerplanar there exist no two edges in $E_1$ alternating along $C$.
Hence, there are no planarity constraints for~$G_1$.

We now define an auxiliary graph $H$ with vertex set $E_2\setminus E_1$ and edges corresponding to pairs of edges alternating along $C$; see Fig~\ref{fig:cycle-algorithm-a}. W.l.o.g. we may assume that $H$ is bipartite, since $G_2$ would not meet the planarity constraints otherwise~\cite{ap-oigs-61}. Let $\cal B$ be the set of connected components of $H$, and for each component $B\in {\cal B}$, fix a partition $B_1,B_2$ of $B$ into independent sets (possibly $B_2=\emptyset$ in case of a singleton $B$). Note that in any inside/outside assignment of the exclusive edges of $G_2$ that meets the planarity constraints, for every $B\in {\cal B}$, all edges of $B_1$ lie in one side of $C$ and all edges of $B_2$ lie in the other side.

Draw the cycle $C$ as a circle in the plane. For a component $B \in {\cal B}$, let $P_B$ be the polygon inscribed into $C$ whose corners are the endvertices in $V$ of the edges in $E_2$ corresponding to the vertices of $B$; refer to Fig.~\ref{fig:cycle-algorithm-b}. If $B$ only contains one vertex (i.e., one edge of $G_2$), we consider the digon $P_B$ as the straight-line segment connecting the vertices of this edge. If $B$ has at least two vertices, we let $P_B$ be open along its sides, i.e. it will contain the corners and all inner points (in Fig.~\ref{fig:cycle-algorithm-b} we depict this by making the sides of $P_B$ slightly concave). One can easily show that for any two components $B,D\in {\cal B}$, their polygons $P_B,P_D$ may share only some of their corners, but no inner points. Hence the graph $\widetilde{G}$ obtained by placing a vertex $x_B$ inside the polygon $P_B$, for $B \in{\cal B}$, making $x_B$ adjacent to each corner of $P_B$ and adding the edges $E_1$, is planar; see Fig.~\ref{fig:cycle-algorithm-c}.

We construct a formula $\phi$ with variables $x_B$, $B \in {\cal B}$, such that $\phi$ is \nae-satisfiable if and only if \ourinstance{} admits an inside/outside assignment meeting all planarity and orthogonality constraints. The encoding of the truth assignment will be such that $x_B$ is {\sf true} when the edges of $B_1$ are inside $C$ and the edges of $B_2$ are outside, and $x_B$ is {\sf false} if the reverse holds.  Every assignment satisfying the planarity constraints for $G_2$ defines a truth-assignment in the above sense. 

Let $e=(v,w)$ be an exclusive edge of $E_1$ and let $e_v^1,e_v^2$ ($e_w^1,e_w^2$) be the exclusive edges of $E_2$ incident to $v$ (to $w$, respectively); we assume that all such four edges of $E_2$ exist, the other cases being simpler. Let $B(u,i)$ be the component containing the edge $e_u^i$, for $u\in \{v,w\}$ and $i \in \{1,2\}$. Define the literal $\ell_u^i$ to be $x_{B(u,i)}$ if $e_u^i \in B_1(u,i)$ and $\neg x_{B(u,i)}$ if $e_u^i\in B_2(u,i)$. With our interpretation of the truth assignment, an edge $e_u^i$ is inside $C$ if and only if $\ell_u^i$ is {\sf true}. Now, for the assignment to meet the orthogonality constraints, if $\ell_v^1 = \ell_v^2$, say both are {\sf true}, then $e$ must be assigned inside $C$ as well, which would cause a problem if and only if $\ell_w^1 = \ell_w^2 = {\sf false}$. Hence the orthogonality constraints are described by \nae-satisfiability of the clauses $c_e = (\ell_v^1, \ell_v^2, \neg \ell_w^1, \neg \ell_w^2)$, for each $e\in E_1$. 
To reduce to \naesat, we introduce a new variable $x_e$ for each edge $e\in E_1 \setminus E_2$ and replace the clause $c_e$ by two clauses $c'_e=(\ell_v^1, \ell_v^2, x_e)$ and $c''_e=(\neg x_e, \neg \ell_w^1, \neg \ell_w^2)$. 
A planar drawing of the variable--clause graph $G_\phi$ of the resulting formula $\phi$ is obtained from the planar drawing $\widetilde{\Gamma}$ of $\widetilde{G}$ (see Figs.~\ref{fig:cycle-algorithm-c} and~\ref{fig:cycle-algorithm-d}) by (i) placing each variable $x_B$, with $B \in {\cal B}$, on the point where vertex $x_B$ lies in $\widetilde{\Gamma}$, (ii) placing each variable $x_e$, with $e \in E_1$, on any point of edge $e$ in $\widetilde{\Gamma}$, (iii) placing clauses $c'_e$ and $c''_e$, for each edge $e=(v,w) \in E_1$, on the points where vertices $v$ and $w$ lie in $\widetilde{\Gamma}$, respectively, and (iv) drawing the edges of $G_\phi$ as the corresponding edges in $\widetilde{\Gamma}$. This implies that $G_\phi$ is planar and hence we can test the \nae-satisfiability of $\phi$ in polynomial time~\cite{moret-pnp-88,shih_etal:90}.
\end{proof}

The next two lemmas show that we can use Lemma~\ref{le:algorithm-degree-three-outerplanar} to test in polynomial time any instance of \ourproblem{2} such that $G_\cap$ is a cycle and each vertex $v \in V$ has degree at most $3$ in either $G_1$ or $G_2$.

\begin{lemma}\label{le:reduction-degree-three-2-degree-three-outerplanar}
  Let \ourinstance{} be an instance of \ourproblem{2} whose
  shared graph is a cycle and such that $G_1$ has maximum degree $3$. It is possible to construct in
  polynomial time an equivalent instance \ourinstance{*} of
  \ourproblem{2} whose shared graph is a
  cycle and such that $G_1^*$ is outerplanar and has maximum degree~$3$.
\end{lemma}

\begin{sketch}
  We construct an equivalent instance \ourinstance{\prime} of
  \ourproblem{2} such that $G_\cap^\prime$ is a cycle, $G_1^\prime$
  has maximum degree $3$, and the number of pairs of edges in
  $G_1^\prime$ that alternate along $G_\cap^\prime$ is smaller than
  the number of pairs of edges in $G_1$ that alternate along $G_\cap$.
  Repeatedly applying this transformation
  yields an equivalent instance \ourinstance{*} satisfying the
  requirements of the lemma.

  Consider two edges $e=(u,v)$ and $f=(w,z)$ of $G_1$ such that
  $u,w,v,z$ appear in this order along cycle $G_\cap$ and such that
  the path $P_{u,z}$ in $G_\cap$  between $u$ and~$z$ that contains $v$
  and $w$ has minimal length. If $G_1$ is not outerplanar, then the edges $e$
  and $f$ always exist. Fig.~\ref{fig:reduction-deg3-2-outerplanar}
  illustrates the construction of \ourinstance{\prime}.

\begin{figure}[tb]
  \centering
  \includegraphics[scale=0.5,page=1]{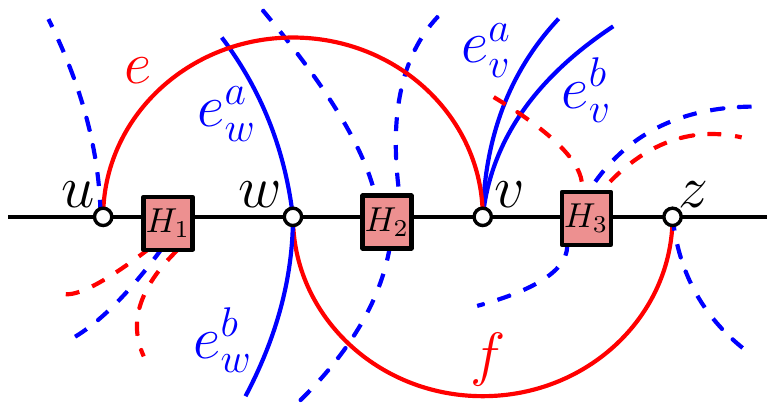}
  \hfil
  \includegraphics[scale=0.5,page=2]{img/red-deg3-outerplanar.pdf}
  \caption[]{Instances (left) \ourinstance{} and (right) \ourinstance{\prime} for the
  proof of Lemma~\ref{le:reduction-degree-three-2-degree-three-outerplanar}.
  Edges of $G_\cap$ ($G'_\cap$) are black. Exclusive edges of $G_1$ ($G'_1$) are red and those of $G_2$ ($G'_2$) are~blue.}
  \label{fig:reduction-deg3-2-outerplanar}
\end{figure}

By the choice of $e$ and $f$, and by the fact that $G_1$ has maximum degree $3$, there is no exclusive edge in $G_1$ with
one endpoint in the set $H_2$ of vertices between $w$ and $v$, and the other one not in $H_2$. Further, observe
that in an \ourdrawing of \ourinstance{\prime} edges $f$ and $f'$ (edges $e$ and $e'$) must be on the same side. Further,~$e$ and
$f$ must be in different sides of $G'_\cap$. It can be concluded that
\ourinstance{\prime} has an \ourdrawing if and only if \ourinstance{}
has an \ourdrawing.
\end{sketch}

The proof of the next lemma is based on the replacement illustrated in Fig.~\ref{fig:red-deg5-to-degree3}.
Afterwards, we combine these results to present the main result of the section.

\begin{figure}[tb]
  \centering
  \includegraphics[scale=0.48,page=1]{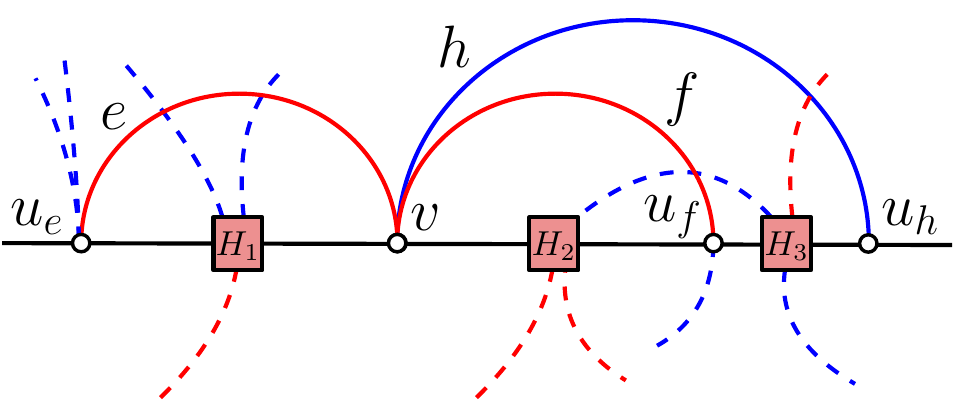}
  \hfill
  \includegraphics[scale=0.48,page=2]{img/red-deg5-to-degree3.pdf}
  \caption[]{
  Illustration of the transformation for the proof of Lemma~\ref{le:reduction-degree-five-2-degree-three} to reduce the number of vertices incident to two exclusive edges in $G_1$.
  Edges $e',f'$ of $G_2$ and $h'$ of $G_1$ (right) take the role of edges $e,f$ of $G_1$ and $h$ of $G_2$ (left), respectively. Thus, the orthogonality constraints at $v'$ are equivalent to those at $v$.
  }
  \label{fig:red-deg5-to-degree3}
\end{figure}

\begin{lemma}\label{le:reduction-degree-five-2-degree-three}
  Let \ourinstance{} be an instance of \ourproblem{2} 
  whose shared graph is a cycle and whose union graph has maximum degree~$5$. It is possible to construct in polynomial 
  time an equivalent instance \ourinstance{*} of \ourproblem{2} whose shared graph is a cycle and such that graph~$G_1^*$ has maximum degree~$3$.
\end{lemma}

\remove{
\begin{sketch}
We describe how to construct an equivalent instance \ourinstance{\prime} of \ourproblem{2} such that $G_\cap^\prime$ is a cycle, $G'_\cup$ has maximum degree~$5$, and the number of degree-$4$ vertices in $G_1^\prime$ is smaller than the number of degree-$4$ vertices in $G_1$.
Repeatedly applying this transformation yields an equivalent instance \ourinstance{*} satisfying the requirements of the lemma.

Let $v \in V$ be a vertex incident to two 
edges $e=(v,u_e), f=(v,u_f) \in E_1$. Assume w.l.o.g. that $u_e$, $v$, and $u_f$ appear in this order along $G_\cap$. 
Suppose that there exists an edge $h=(v,u_h) \in E_2$ incident to $v$, the other
case being simpler. Fig.~\ref{fig:red-deg5-to-degree3}
shows how to construct \ourinstance{\prime} when vertices
$u_e$, $v$, $u_h$, $u_f$ appear in this order along $G_\cap$; the other cases are analogous.
\end{sketch}
}

\begin{theorem}\label{th:algorithm-one-degree-three}
\ourproblem{2} can be solved in polynomial time for instances whose shared graph is a cycle and whose union graph has maximum degree~$5$.
\end{theorem}

\section{Shared Graph is Biconnected}
\label{se:biconnected}

We now study \ourproblem{k} for 
instances whose shared graph is biconnected.  In 
Theorem~\ref{th:decision-biconnected}, we give a polynomial-time Turing 
reduction from instances of \ourproblem{2} whose shared graph is biconnected to 
instances whose shared graph is a cycle. In 
Theorem~\ref{th:biconnected geometric}, 
we give an algorithm that, given a positive instance of \ourproblem{k} such that the shared graph is biconnected together with a SEFE satisfying the orthogonality constraints,
constructs an \ourdrawing with at most three bends per edge.


We start with the Turing reduction, i.e., we develop an algorithm that
takes as input an instance \ourinstance{} of \ourproblem{2} whose
shared graph $G_\cap = G_1 \cap G_2$ is biconnected and produces a set
of $O(n)$ instances \ourinstance{1},\dots,\ourinstance{h} of
\ourproblem{2} whose shared graphs are cycles.  
The output is such that
\ourinstance{} is a positive instance if and only if all instances
\ourinstance{i}, $i=1,\dots,h$, are positive.  The reduction
is based on the {\sc SEFE} testing algorithm for instances whose shared
graph is biconnected by Bläsius et
al.~\cite{bkr-seeor-13,bkr-seeor-15}, which can be seen as a
generalized and unrooted version of the one by Angelini et
al.~\cite{adfpr-tsetgibgt-11}.

We first describe a preprocessing step. 
Afterwards, we give an outline of the approach of Bläsius et
al.~\cite{bkr-seeor-15} and present the Turing reduction in two steps. We assume familiarity with SPQR-trees~\cite{dt-opt-96,dt-omtcws-96}; for formal definitions see Appendix~\ref{app:definitions}.

\begin{lemma}
  \label{le:simple-attachments}
  Let \ourinstance{} be an instance of \ourproblem{2} whose shared graph 
  is biconnected. It is possible to construct in polynomial time an equivalent instance \ourinstance{*}
  whose shared graph is biconnected and such that each endpoint of an
  exclusive edge has degree~$2$ in the shared graph.
\end{lemma}

We continue with a brief outline of the algorithm by Bläsius et
al.~\cite{bkr-seeor-15}.  First, the algorithm computes the SPQR-tree
$\mathcal T$ of the shared graph.  To avoid special cases, $\mathcal
T$ is augmented by adding S-nodes with only two virtual edges such
that each P-node and each R-node is adjacent only to S-nodes and
Q-nodes.
\remove{ ; in particular, we introduce such an S-node
  between any two adjacent R-nodes and any adjacent P- and R-nodes
  (recall that two P-nodes cannot be adjacent).  }
Then, necessary conditions on the embeddings of P-nodes and R-nodes
are fixed up to a flip following some necessary conditions.
Afterwards, by traversing all S-nodes, a global 2SAT formula is
produced whose satisfying assignments correspond to choices of the
flips that result in a SEFE.  We refine this approach and show that we
can choose the flips independently for each S-node, which
allows us to reduce each of them to a separate instance, whose shared
graph is a cycle.

We now describe the algorithm of Bläsius et al.~\cite{bkr-seeor-15} in
more detail.  Consider a node $\mu$ of $\mathcal T$.  A \emph{part} of
$\skel(\mu)$ is either a vertex of $\skel(\mu)$ or a virtual edge of
$\skel(\mu)$, which represents a subgraph of $G$.  An exclusive edge
$e$ has an \emph{attachment} in a part $x$ of $\skel(\mu)$ if $x$ is a
vertex that is an endpoint of $e$ or if $x$ is a virtual edge whose
corresponding subgraph contains an endpoint of $e$.  An exclusive edge
$e$ of $G_1$ or of $G_2$ is \emph{important} for $\mu$ if its
endpoints are in different parts of $\skel(\mu)$.  It is not hard to
see that, to obtain a SEFE, the embedding of the skeleton $\skel(\mu)$
of each node $\mu$ has to be chosen such that for each exclusive edge
$e$ the parts containing the attachments of $e$ share a face.  It can
be shown that any embedding choice for P-nodes and R-nodes that
satisfies these conditions can, after possibly flipping it, be used to
obtain a SEFE~\cite[Theorem~1]{adfpr-tsetgibgt-11}. The proof does not
modify the order of exclusive edges around degree-$2$ vertices of
$G_\cap$, and therefore applies to \ourproblem{2} as well.



Now let $\mu$ be an S-node.
\remove{For simplicity, we direct the edges of
$\skel(\mu)$ so that we can speak about a left and right side of
$\skel(\mu)$.}
Let $\eps$ be a virtual edge of $\skel(\mu)$, $G_\eps$ be the subgraph
represented by $\eps$, and $\nu$ be the corresponding neighbor of
$\mu$ in the SPQR-tree of $G$.  An \emph{attachment} of $\nu$ with
respect to $\mu$ is an interior vertex of $G_\eps$ that is incident to
an important edge $e$ for $\mu$.  If $\nu$ has such an attachment, then it is a P- or R-node.
\remove{Since the edge $e$ is important, it has to be embedded inside
  one of the faces of $\skel(\mu)$. Thus it enforces that, in the
  embedding of the subgraph $G_\eps$ represented by $\eps$, its
  attachment $x$ must be incident to a face incident to the virtual
  edge $\twin(\eps)$ of $\skel(\nu)$ representing $\mu$, i.e., it must
  be incident to the outer face of $G_\eps$ when it is rooted at
  $\twin(\eps)$.  Since $G_\eps$ has an internal vertex, $\nu$ cannot
  be a Q-node, and since no two S-nodes are adjacent in the SPQR-tree
  of $G$, $\nu$ is either a P-node or an R-node, and thus $G_\eps$ is
  biconnected.}  It is a necessary condition on the embedding of
$G_\eps$ that each attachment $x$ with respect to $\mu$ must be
incident to a face incident to the virtual edge $\twin(\eps)$ of
$\skel(\nu)$ representing $\mu$, and that their clockwise circular
order together with the poles of $\eps$ is fixed up to
reversal~\cite[Lemma 8]{bkr-seeor-15}.

For the purpose of avoiding crossings in $\skel(\mu)$, we can thus
replace each virtual edge $\eps$ that does not represent a Q-node by a
cycle $C_\eps$ containing the attachments of $\eps$ with respect to
$\mu$ and the poles of $\eps$ in the order $O_\eps$.  We keep only the
important edges of $\mu$.  Altogether this results in an instance
\ourinstance{\mu} of {\sc SEFE} modeling the requirements for
$\skel(\mu)$; see Figs.~\ref{fig:biconnected-reduction-a}
and~\ref{fig:biconnected-reduction-b}.

\begin{figure}[tb!]
  \centering
  \subfigure[]{\includegraphics[page=1]{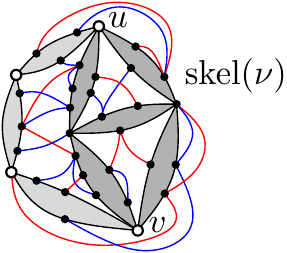}\label{fig:biconnected-reduction-a}}
  \hfil
  \subfigure[]{\includegraphics[page=2]{img/biconnected-reduction}\label{fig:biconnected-reduction-b}}
  \hfil
  \subfigure[]{\includegraphics[page=3]{img/biconnected-reduction}\label{fig:biconnected-reduction-c}}  
  \label{fig:biconnected-reduction}
  \caption{(a) Skeleton of an S-node $\mu$ in which the R-node $\nu$ corresponding to the virtual edge $\eps = (u,v)$ is expanded to show its skeleton. (b) Replacing $\eps$ with cycle $C_\eps$. (c) Replacing $C_\eps$ with path $P_\eps$; vertices $a_1,a_2,x_1,\dots,x_4,b_1,b_2$ are green boxes.}
\end{figure}

\begin{lemma}
  \label{le:s-nodes-independent}
  Let \ourinstance{} be an instance of \ourproblem{2} whose
    shared graph is biconnected.  Then
    \ourinstance{} admits an \ourdrawing if and only if all
    instances \ourinstance{\mu} admit an \ourdrawing.
\end{lemma}

Next, we transform a given instance \ourinstance{\mu} of
\ourproblem{2} as above into an equivalent instance $\langle
\overline{G_1^{\mu}}, \overline{G_2^{\mu}}\rangle$ whose shared graph
is a cycle.  Let $C_{\eps_i}$ be the cycles
corresponding to the neighbor $\nu_i$, $i=1,\dots,k$ of $\mu$ in
\ourinstance{\mu}.  To obtain the instance $\langle
\overline{G_1^{\mu}}, \overline{G_2^{\mu}}\rangle$, we replace each
cycle $C_{\eps_i}$ with poles $u$ and $v$ by a path $P_{\eps_i}$ from
$u$ to~$v$ that first contains two special vertices $a_1,a_2$ followed
by the clockwise path from $u$ to $v$ (excluding the endpoints), then
four special vertices $x_1,\dots,x_4$, then the counterclockwise path
from $u$ to $v$ (excluding the endpoints), and finally two special
vertices $b_1,b_2$ followed by $v$.  In addition to the existing
exclusive edges (note that we do not remove any vertices), we add to
$G_1$ the exclusive edges $(a_2,x_3)$, $(x_1,x_3)$, $(x_2,x_4)$,
$(x_2,b_1)$, and to $G_2$ the exclusive edges $(a_1,x_3)$ and
$(x_2,b_2)$ to $G_2$; see Fig.~\ref{fig:biconnected-reduction-c}.

The above reduction together with the next lemma implies the main result.

\begin{lemma}
  \label{le:cycle-to-path-gadget}
  \ourinstance{\mu} admits an \ourdrawing if and only if
  $\langle \overline{G_1^{\mu}}, \overline{G_2^{\mu}}\rangle$ does.
\end{lemma}

\begin{theorem}\label{th:decision-biconnected}
  
  \ourproblem{2} when
  the shared graph is biconnected is polynomial-time Turing reducible to \ourproblem{2} when the shared
  graph is a cycle. Also, the reduction does not increase the maximum degree of the union graph.
\end{theorem}

\remove{
\begin{proof}
  The reduction follows immediately from
  Lemmas~\ref{le:simple-attachments}, \ref{le:s-nodes-independent},
  and~\ref{le:cycle-to-path-gadget}.  Each of the steps of the
  construction can clearly be carried out in polynomial time. Observe that, none of the reductions increases the maximum degree of the union graph.
\end{proof}
}


\begin{corollary}
\ourproblem{2} can be solved in polynomial time for instances 
whose shared graph is biconnected and whose union graph has maximum degree~$5$.
\end{corollary}


Observe that, from the previous results it is not hard to also obtain a SEFE satisfying the orthogonality constraints, if it exists. We show how to construct an orthogonal geometric realizations of such a SEFE.

\begin{theorem}\label{th:biconnected geometric}
  Let \ourinstancek{} be a positive instance of \ourproblem{k} whose shared graph is biconnected. Then,
  there exists an \ourdrawing \oursolutionk{} of \ourinstancek{} in
  which every edge has at most three bends.
\end{theorem}

\begin{sketch}
  We assume that a SEFE satisfying the orthogonality
  constraints is given.
  %
  %
  \begin{figure}[tb!]
\centering
\hfil
  \subfigure[\label{FIG:geo_v1}around $v_1$]{\includegraphics[scale=0.9,page=1]{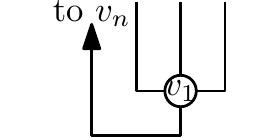}}
  \hfil
  \subfigure[\label{FIG:geo_vi}around $v_2,\dots,v_{n-1}$]{\includegraphics[scale=0.9,page=4]{img/geo.pdf}}
  \hfil
  \subfigure[\label{FIG:geo_vn}around $v_n$]{\includegraphics[scale=0.9,page=5]{img/geo.pdf}}

\caption{\label{FIG:geo}Constructing a drawing with at most three bends per edge}
\end{figure}
  We adopt the method of Biedl and Kant~\cite{biedl/kant:98}. 
  We draw the vertices with increasing y-coordinates with respect to
  an $s$-$t$-ordering~\cite{brandes:02} $v_1,\dots,v_n$ on
  the shared graph.
  We choose the face to
  the left of $(v_1,v_n)$ as the outer face of the union graph.
  %
  The edges will bend at most on y-coordinates near their incident
  vertices and are drawn vertically otherwise. Fig.~\ref{FIG:geo}
  indicates, how the ports are assigned. We
  make sure that an edge may only leave a vertex to the bottom if it is
  incident to $v_n$ or to a neighbor with a lower index.
  Thus, there are exactly three
  bends on $\{v_1,v_n\}$. Any other
  edge $\{v_i,v_j\}$, $1 \leq i < j \leq n$ has at most one bend around $v_i$ 
  and at most two bends around $v_j$.
\end{sketch}

\section{Conclusions and Future Work}\label{se:conclusions}

In this work we introduced and studied the problem \ourproblem{k} of realizing a {\sc SEFE} in the orthogonal drawing style. 
While the problem is already NP-hard even for instances that can be efficiently tested for a {\sc SEFE}, we presented a polynomial-time testing algorithm for instances consisting of two graphs whose shared graph is biconnected and whose union graph has maximum degree~$5$.
We have also shown that any positive instance whose shared graph is biconnected can be realized with at most three bends per edge.

We conclude the paper by presenting a lemma that, together with Theorem~\ref{th:decision-biconnected}, shows that it suffices to only focus on
a restricted family of instances to solve the problem for all instances whose
shared graph is biconnected.

\begin{lemma}
\label{le:reduction-general-2-degree-3,5}
Let \ourinstance{} be an instance of \ourproblem{2} whose shared graph $G_\cap$ is a cycle.
An equivalent instance \ourinstance{*} of \ourproblem{2} such that 
(i) the shared graph $G_\cap^*$ is a cycle,
(ii) graph $G_1^*$ is outerplanar, and
(iii) no two degree-$4$ vertices in $G_1^*$ are adjacent, can be constructed in polynomial time.
\end{lemma}

\clearpage

\bibliographystyle{abbrv}
\bibliography{bibliography}

\begin{thebibliography}{10}

\bibitem{adfpr-tsetgibgt-11}
P.~Angelini, G.~D. Battista, F.~Frati, M.~Patrignani, and I.~Rutter.
\newblock Testing the simultaneous embeddability of two graphs whose
  intersection is a biconnected or a connected graph.
\newblock {\em J. Discrete Algorithms}, 14:150--172, 2012.

\bibitem{addfpr-blp-15}
P.~Angelini, G.~{Da Lozzo}, G.~{Di Battista}, F.~Frati, M.~Patrignani, and
  I.~Rutter.
\newblock Beyond level planarity.
\newblock In M.~N{\"{o}}llenburg and Y.~Hu, editors, {\em GD '16}, LNCS.
  Springer, 2017.
\newblock To appear.

\bibitem{adn-aspbep-15}
P.~Angelini, G.~D. Lozzo, and D.~Neuwirth.
\newblock Advancements on {SEFE} and partitioned book embedding problems.
\newblock {\em Theor. Comput. Sci.}, 575:71--89, 2015.

\bibitem{abks-gracsdg-13}
E.~N. Argyriou, M.~A. Bekos, M.~Kaufmann, and A.~Symvonis.
\newblock Geometric {RAC} simultaneous drawings of graphs.
\newblock {\em J. Graph Algorithms Appl.}, 17(1):11--34, 2013.

\bibitem{ap-oigs-61}
L.~Auslander and S.~V. Parter.
\newblock On embedding graphs in the sphere.
\newblock {\em J. Math. Mech.}, 10(3):517--523, 1961.

\bibitem{dt-omtcws-96}
G.~D. Battista and R.~Tamassia.
\newblock On-line maintenance of triconnected components with {SPQR}-trees.
\newblock {\em Algorithmica}, 15(4):302--318, 1996.

\bibitem{dt-opt-96}
G.~D. Battista and R.~Tamassia.
\newblock On-line planarity testing.
\newblock {\em {SIAM} J. Comput.}, 25(5):956--997, 1996.

\bibitem{bdkw-sdpgrac-16}
M.~A. Bekos, T.~C. van Dijk, P.~Kindermann, and A.~Wolff.
\newblock Simultaneous drawing of planar graphs with right-angle crossings and
  few bends.
\newblock {\em J. Graph Algorithms Appl.}, 20(1):133--158, 2016.

\bibitem{biedl/kant:98}
T.~Biedl and G.~Kant.
\newblock A better heuristic for orthogonal graph drawings.
\newblock {\em Comput. Geom.}, 9(3):159--180, 1998.

\bibitem{bkr-seeor-13}
T.~Bl{\"a}sius, A.~Karrer, and I.~Rutter.
\newblock Simultaneous embedding: Edge orderings, relative positions,
  cutvertices.
\newblock In S.~K. Wismath and A.~Wolff, editors, {\em GD'13}, volume 8242 of
  {\em LNCS}, pages 220--231. Springer, 2013.

\bibitem{bkr-seeor-15}
T.~Bl{\"{a}}sius, A.~Karrer, and I.~Rutter.
\newblock Simultaneous embedding: Edge orderings, relative positions,
  cutvertices.
\newblock {\em ArXiv e-prints}, abs/1506.05715, 2015.

\bibitem{bkr-sepg-12}
T.~Bl{\"{a}}sius, S.~G. Kobourov, and I.~Rutter.
\newblock Simultaneous embedding of planar graphs.
\newblock In R.~Tamassia, editor, {\em Handbook of Graph Drawing and
  Visualization}. CRC Press, 2013.

\bibitem{br-drpse-15}
T.~Bl{\"{a}}sius and I.~Rutter.
\newblock Disconnectivity and relative positions in simultaneous embeddings.
\newblock {\em Comput. Geom.}, 48(6):459--478, 2015.

\bibitem{br-spqoacep-13}
T.~Bl{\"{a}}sius and I.~Rutter.
\newblock Simultaneous {PQ}-ordering with applications to constrained embedding
  problems.
\newblock {\em {ACM} Trans. Alg.}, 12(2):16, 2016.

\bibitem{brandes:02}
U.~Brandes.
\newblock Eager st-ordering.
\newblock In R.~H. M{\"{o}}hring and R.~Raman, editors, {\em ESA'02}, volume
  2461 of {\em LNCS}, pages 247--256. Springer, 2002.

\bibitem{dt-ogasp-90}
G.~{Di Battista} and R.~Tamassia.
\newblock On-line graph algorithms with {SPQR}-trees.
\newblock In M.~S. Paterson, editor, {\em ICALP'90}, volume 443 of {\em LNCS},
  pages 598--611. Springer, 1990.

\bibitem{egjpss-sgge-07}
A.~Estrella{-}Balderrama, E.~Gassner, M.~J{\"{u}}nger, M.~Percan, M.~Schaefer,
  and M.~Schulz.
\newblock Simultaneous geometric graph embeddings.
\newblock In S.~Hong, T.~Nishizeki, and W.~Quan, editors, {\em GD'07}, volume
  4875 of {\em LNCS}, pages 280--290. Springer, 2007.

\bibitem{gm-lis-00}
C.~Gutwenger and P.~Mutzel.
\newblock A linear time implementation of {SPQR}-trees.
\newblock In J.~Marks, editor, {\em GD'00}, volume 1984 of {\em LNCS}, pages
  77--90. Springer, 2000.

\bibitem{hjl-tspcg2c-10}
B.~Haeupler, K.~R. Jampani, and A.~Lubiw.
\newblock Testing simultaneous planarity when the common graph is 2-connected.
\newblock {\em J. Graph Algorithms Appl.}, 17(3):147--171, 2013.

\bibitem{simultaneous_interval_graphs}
K.~R. Jampani and A.~Lubiw.
\newblock Simultaneous interval graphs.
\newblock In O.~Cheong, K.~Chwa, and K.~Park, editors, {\em ISAAC'10}, volume
  6506 of {\em LNCS}, pages 206--217, 2010.

\bibitem{jl-srpcc-12}
K.~R. Jampani and A.~Lubiw.
\newblock The simultaneous representation problem for chordal, comparability
  and permutation graphs.
\newblock {\em J. Graph Algorithms Appl.}, 16(2):283--315, 2012.

\bibitem{js-igsefe-09}
M.~J{\"{u}}nger and M.~Schulz.
\newblock Intersection graphs in simultaneous embedding with fixed edges.
\newblock {\em J. Graph Algorithms Appl.}, 13(2):205--218, 2009.

\bibitem{moret-pnp-88}
B.~M.~E. Moret.
\newblock Planar {NAE3SAT} is in {P}.
\newblock {\em ACM SIGACT News}, 19(2):51--54, 1988.

\bibitem{moret-toc-98}
B.~M.~E. Moret.
\newblock {\em Theory of Computation}.
\newblock Addison-Wesley-Longman, 1998.

\bibitem{Pap07}
C.~H. Papadimitriou.
\newblock {\em Computational complexity}.
\newblock Academic Internet Publ., 2007.

\bibitem{s-ttphtpv-13}
M.~Schaefer.
\newblock Toward a theory of planarity: {H}anani--{T}utte and planarity
  variants.
\newblock {\em J. Graph Algorithms Appl.}, 17(4):367--440, 2013.

\bibitem{Schaefer-tcsp-78}
T.~J. Schaefer.
\newblock The complexity of satisfiability problems.
\newblock In R.~J. Lipton, W.~A. Burkhard, W.~J. Savitch, E.~P. Friedman, and
  A.~V. Aho, editors, {\em Proceedings of the 10th Annual ACM Symposium on
  Theory of Computing (STOC'78)}, pages 216--226. {ACM}, 1978.

\bibitem{shih_etal:90}
W.~Shih, S.~Wu, and Y.~Kuo.
\newblock Unifying maximum cut and minimum cut of a planar graph.
\newblock {\em {IEEE} Trans. Computers}, 39(5):694--697, 1990.

\bibitem{t-eggmdb-87}
R.~Tamassia.
\newblock On embedding a graph in the grid with the minimum number of bends.
\newblock {\em {SIAM} J. Comput.}, 16(3):421--444, 1987.

\end{thebibliography}

\clearpage
\appendix

\noindent{\Large\bfseries Appendix}
\section{Definitions for the appendix}\label{app:definitions}

In Section~\ref{se:embedding-constraints} we already discussed how to assign the
exclusive edges to either of the two sides of $C$.
We formalise this assignment by means of a function $A: \bigcup^k_{i=1} E_i
\setminus E(G_\cap) \rightarrow \{l,r\}$, where $A(e)=l$ (resp. $A(e)=r$) if
edge $e$ lies to the left (resp. to the right) of $C$, according to an arbitrary
orientation of $C$.

\subsubsection*{Connectivity and SPQR-trees.} 
A graph $G = (V,E)$ is \emph{connected} if there is a path between any
two vertices.  A \emph{cutvertex} is a vertex whose removal
disconnects the graph.  A separating pair $\{u,v\}$ is a pair of
vertices whose removal disconnects the graph.  A connected graph is
\emph{biconnected} if it does not have a cutvertex and a biconnected
graph is \emph{3-connected} if it does not have a separating pair.

We consider $st$-graphs with two special \emph{pole} vertices $s$ and
$t$.  The family of $st$-graphs can be constructed in a fashion very
similar to series-parallel graphs.  Namely, an edge $st$ is an
$st$-graph with poles $s$ and $t$.  Now let $G_i$ be an $st$-graph
with poles $s_i,t_i$ for $i=1,\dots,k$ and let $H$ be a planar graph
with two designated adjacent vertices $s$ and $t$ and $k+1$ edges $st,
e_1,\dots,e_k$.  We call $H$ the \emph{skeleton} of the composition
and its edges are called \emph{virtual edges}; the edge $st$ is the
\emph{parent edge} and $s$ and $t$ are the poles of the skeleton $H$.
To compose the $G_i$'s into an $st$-graph with poles $s$ and $t$, we
remove the edge $st$ from $H$ and replace each $e_i$ by $G_i$ for
$i=1,\dots,k$ by removing $e_i$ and identifying the poles of $G_i$
with the endpoints of $e_i$.  In fact, we only allow three types of
compositions: in a \emph{series composition} the skeleton $H$ is a
cycle of length~$k+1$, in a parallel composition $H$ consists of two
vertices connected by $k+1$ parallel edges, and in a \emph{rigid
  composition} $H$ is 3-connected.

For every biconnected planar graph $G$ with an edge $st$, the graph
$G-st$ is an $st$-graph with poles $s$ and $t$~\cite{dt-ogasp-90}.
Much in the same way as series-parallel graphs, the $st$-graph $G-st$
gives rise to a (de-)composition tree~$\mathcal T$ describing how it
can be obtained from single edges.  The nodes of $\mathcal T$
corresponding to edges, series, parallel, and rigid compositions of
the graph are \emph{Q-, S-, P-,} and \emph{R-nodes}, respectively.  To
obtain a composition tree for~$G$, we add an additional root Q-node
representing the edge $st$.  We associate with each node~$\mu$ the
skeleton of the composition and denote it by $\skel(\mu)$.  For a
Q-node~$\mu$, the skeleton consists of the two endpoints of the edge
represented by~$\mu$ and one real and one virtual edge between them
representing the rest of the graph.  For a node~$\mu$ of $\mathcal T$,
the \emph{pertinent graph} $\pert(\mu)$ is the subgraph represented by
the subtree with root~$\mu$.  For a virtual edge $\eps$ of a
skeleton~$\skel(\mu)$, the \emph{expansion graph} of $\eps$ is the
pertinent graph $\pert(\mu')$ of the neighbor $\mu'$ corresponding to
$\eps$ when considering $\mathcal T$ rooted at $\mu$.

The \emph{SPQR-tree} of $G$ with respect to the edge $st$, originally
introduced by Di Battista and Tamassia~\cite{dt-ogasp-90}, is the
(unique) smallest decomposition tree~$\mathcal T$ for $G$.  Using a
different edge $s't'$ of $G$ and a composition of $G-s't'$ corresponds
to rerooting $\mathcal T$ at the node representing $s't'$.  It thus
makes sense to say that $\mathcal T$ is the SPQR-tree of $G$.  The
SPQR-tree of $G$ has size linear in $G$ and can be computed in linear
time~\cite{gm-lis-00}.  Planar embeddings of $G$ correspond
bijectively to planar embeddings of all skeletons of $\mathcal T$; the
choices are the orderings of the parallel edges in P-nodes and the
embeddings of the R-node skeletons, which are unique up to a flip.
When considering rooted SPQR-trees, we assume that the embedding of
$G$ is such that the root edge is incident to the outer face, which is
equivalent to the parent edge being incident to the outer face in each
skeleton.
We remark that in a planar embedding of $G$, the poles of any node
$\mu$ of $\mathcal{T}$ are incident to the outer face of
$pert(\mu)$. Hence, in the following we only consider embeddings of
the pertinent graphs with their poles lying on the same face.

\section{Omitted or Sketched Proofs from Section~\ref{se:preliminaries}}\label{app:preliminaries}

\rephrase{Theorem}{\ref{th:general-characterization}}{
An instance \ourinstancek{} of \ourproblem{k} has an \ourdrawing if and only if it admits a SEFE satisfying the orthogonality constraints.
}

\begin{proof} 
For the {\em if} part, let $\cal E$ be the embedding of $G_\cap$ determined by the SEFE $\langle \mathcal E_1,\dots \mathcal E_k \rangle$ of \ourinstancek{}. Observe that the orthogonality constraints at each vertex define (i) whether a degree~$2$ vertex of $G_\cap$ has to be drawn straight or bent, and (ii) which face incident to a degree~$3$ vertex of $G_\cap$ has to be assigned the $180^\circ$ angle. It is not hard to see that a planar orthogonal drawing $\Gamma$ of $G_\cap$ in which the embedding of $G_\cap$ is $\cal E$ satisfying such requirements can be constructed. We draw the exclusive edges in each $E_i$ as orthogonal polylines in $\Gamma$ inside the face of $\cal E$ determined by the SEFE.
The fact that the exclusive edges of each $E_i$ can be drawn in $\Gamma$ without introducing any crossings descends from the fact that $\mathcal E_i$ is a planar embedding of $G_i$.

 For the {\em only if} part, let~$v$ be a vertex in~$G_\cap$ such that the orthogonality constraints are not satisfied at $v$.
  If $v$ has exactly two neighbors $u$ and $w$ in~$G_\cap$, then we need to assign a port to two exclusive edges of the same graph (one for each of these edges)
  on one side of the path~$uvw$ and a port to at least one exclusive edge on the other side of the path~$uvw$.
 If $v$ has degree~3 in~$G_\cap$, then we need to assign a port to an exclusive edge between a pair of edges of $G_\cap$ and a port to an exclusive edge between a different pair of edges of $G_\cap$. 
Hence, in both cases  we need at least five ports, which is not possible on the grid.  
\end{proof}

\section{Omitted or Sketched Proofs from Section~\ref{se:cycle-hardness}}\label{app:cycle-hardness}

\rephrase{Theorem}{\ref{th:hardness-hamiltonian-cycle-three-colors}}{
\ourproblem{3} is NP-complete, even for instances $\langle G_1=(V,E_1), G_2=(V,E_2), G_3=(V,E_3) \rangle$ with sunflower intersection in which (i) the shared graph $G_\cap= (V, E_1 \cap E_2 \cap E_3)$ is a cycle and (ii) $G_1$ and $G_2$ are outerplanar graphs with maximum degree $3$.}

\begin{proof}
The membership in NP directly follows from Theorem~\ref{th:characterization}, since an assignment $A$, which is a certificate for our problem, can be easily verified in polynomial time to satisfy all the planarity and the orthogonality constraints.

To prove that the problem is NP-hard, we show a reduction from the NP-complete
problem {\sc Positive Exactly-Three} \naesat~\cite{moret-toc-98}, which is the
variant of \naesat in which each clause consists of exactly three unnegated
literals. See Fig.~\ref{fig:hardness-3-graphs}.


Let $x_1,x_2,\dots,x_n$ be the variables and let $c_1,c_2,\dots,c_m$ be the clauses of a $3$-CNF formula $\phi$ of {\sc Positive Exactly-Three} \naesat. We show how to construct an equivalent instance $\langle G_1=(V,E_1),G_2=(V,E_2),G_3=(V,E_3)\rangle$ of \ourproblem{3}; refer to Fig.~\ref{fig:hardness-3-graphs-b}. Assume, without loss of generality, that the literals in each clause $c_j=(x^j_a,x^j_b,x^j_c)$ are such that $a > b > c$, if $j$ is odd, and $a < b < c$, otherwise.

A \emph{variable-clause gadget} $V^j_i$ for a variable $x_i$ belonging to a
clause $c_j$ is a subgraph of $G_\cup$ defined as follows. Gadget $V^i_j$
contains a path $(s^j_i,u^j_i,w^j_i,v^j_i,z^j_i,r^j_i,t^j_i)$ belonging to
$G_\cap$, and edges $\{u^j_i,v^j_i\}$ and $\{w^j_i,z^j_i\}$ belonging to $E_3$;
see Fig.~\ref{fig:hardness-3-graphs-a}.

The \emph{clause gadget} $C^j$ for a clause $c_j$ is a subgraph of $G_\cup$ defined as follows.
Gadget $C^j$ contains a path
$(s^j,\alpha^j,y^j_a,\beta^j,y^j_b,d^j_1,\dots,d^j_6,\gamma^j,y^j_c,\delta^j,t^j)$
belonging to $G_\cap$, and edges $\{\alpha^j,\beta^j\}$, $\{\beta^j,\gamma^j\}$,
$\{\gamma^j,\delta^j\}$, $\{d^j_1,d^j_3\}$, $\{d^j_2,d^j_4\}$, $\{d^j_3,d^j_5\}$,
$\{d^j_4,d^j_6\}$ belonging to $E_3$, and edges $\{\beta^j,d^j_3\}$ and
$\{d^j_4,\gamma^j\}$ belonging to $E_1$; see Fig.~\ref{fig:hardness-3-graphs-a}.

Initialize $G_\cup$ to the union of $V^j_i$, for $i=1,\dots,n$ and $j=1,\dots,m$, and of $C^j$,
for $j=1,\dots,m$. Then, for $j=1,\dots,m$ and for $i=1,\dots,n-1$ identify vertex $t^j_i$ with vertex $s^j_{i+1}$, if $j$ is odd, or identify vertex $t^j_{i+1}$ with vertex $s^j_{i}$, otherwise. Further, for $j=1,\dots,m$ (where $m+1=1$), identify vertex $t^j_n$ with vertex $s^j$ and vertex $t^j$ with vertex $s^{j+1}_n$, if $j$ is odd, or identify vertex $t^j_1$ with vertex $s^j$ and vertex $t^j$ with vertex $s^{j+1}_1$, otherwise.

To complete the construction of $\langle G_1,G_2,G_3\rangle$ we add to $G_\cup$
exclusive edges as follows. For $i=1,\dots,n$ and for $j=1,\dots,m-1$,
we add an edge $\{w^j_i,w^{j+1}_i\}$ to $E_2$, if $j$ is odd, or to
$E_1$, otherwise. We call these edges {\em transmission edges}.
Further, for $i=1,\dots,n$ and for $j=1,\dots,m$, we add an edge
$\{w^j_i,y^j_i\}$, if $x_i \in c_j$, or an edge $\{w^j_i,r^j_i\}$, otherwise.

Clearly, the construction of instance $\langle G_1,G_2,G_3\rangle$ can be
completed in polynomial time.

Graph $G_\cap$ is a cycle, as we already observed. Also, the transmission edges in $E_1$ (in $E_2$) do not alternate along $G_\cap$, since the variable-clause gadgets appear along $G_\cap$ in 
the order $V_1^j,\dots,V_n^j$, if $j$ is odd, or in the order $V_n^j,\dots,V_1^j$, otherwise. Also, no transmission edge in $E_1$ alternates with edges $(\beta^j,d_3)$ and $(d_4,\gamma^j)$, for any $j$, and such edges do not alternate with each other by construction. Hence, $G_1$ and $G_2$ are outerplanar. The fact that $G_1$ and $G_2$ have maximum degree $3$ also directly follows from the construction.

Given a positive instance $\phi$ of {\sc Positive Exactly-Three} \naesat, we
show that $\langle G_1,G_2,G_3\rangle$ is a positive instance of \ourproblem{3}.
Given a satisfying truth assignment
$T: X \rightarrow \{\true,\false\}$ where $X$ denotes the set of
variables in $\phi$, we construct an assignment $A$ of the exclusive edges of
$E_1 \cup E_2 \cup E_3$ to the two sides of $G_\cap$ satisfying all the
planarity and the orthogonality constraints.

For $i=1,\dots,n$ and for $j=1,\dots,m$, we set $A(e)=l$, for each exclusive
edge $e \in E_1 \cup E_2 \cup E_3$ incident to $w^j_i$, if $T(x_i)=\true$, or
$A(e)=r$, otherwise.
For $i=1,\dots,n$ and for $j=1,\dots,m$, we set $A(u_i^j,v_i^j)=r$, if
$T(x_i)=\true$, or $A(u_i^j,v_i^j)=l$, otherwise.
For each clause $c_j=(x^j_a,x^j_b,x^j_c)$,
we set $A(\alpha^j,\beta^j)=l$, if $T(x_a)=\false$, or $A(\alpha^j,\beta^j)=r$,
otherwise;
we set $A(\beta^j,\gamma^j)=l$, if $T(x_b)=\false$, or $A(\beta^j,\gamma^j)=r$,
otherwise;
and we set $A(\gamma^j,\delta^j)=l$, if $T(x_c)=\false$, or
$A(\gamma^j,\delta^j)=r$, otherwise.
Finally, for each clause $c_j=(x^j_a,x^j_b,x^j_c)$, consider the literal
$x_\circ$ with $\circ \in \{a,c\}$ such that $T(x_\circ)=T(x_b)$, if any,
otherwise let $x_\circ=x_a$.
Suppose that $x_\circ=x_a$; set $A(d^j_1,d^j_3)=A(d^j_3,d^j_5)=A(\beta^j,d^j_3)=r$ and set $A(d^j_2,d^j_4)=A(d^j_4,d^j_6)=A(d^j_4,\gamma^j)=l$, if $T(x_\circ)=\false$, or set $A(d^j_1,d^j_3)=A(d^j_3,d^j_5)=A(\beta^j,d^j_3)=l$ and set $A(d^j_2,d^j_4)=A(d^j_4,d^j_6)=A(d^j_4,\gamma^j)=r$, otherwise.
Suppose that $x_\circ=x_c$; set $A(d^j_1,d^j_3)=A(d^j_3,d^j_5)=A(\beta^j,d^j_3)=l$ and set $A(d^j_2,d^j_4)=A(d^j_4,d^j_6)=A(d^j_4,\gamma^j)=r$, if $T(x_\circ)=\false$, or set $A(d^j_1,d^j_3)=A(d^j_3,d^j_5)=A(\beta^j,d^j_3)=l$ and set $A(d^j_2,d^j_4)=A(d^j_4,d^j_6)=A(d^j_4,\gamma^j)=r$, otherwise.


We show that $A$ satisfies the planarity constraints. First observe that the
planarity constraints for the edges in $E_1$ and $E_2$ are trivially satisfied
by $A$ since $G_1$ and $G_2$ are outerplanar. As for the edges in $E_3$, we have
that the only pairs of edges that alternate along $G_\cap$ are $\langle
(u_i^j,v_i^j),(w_i^j,z_i^j)\rangle$, for $i=1,\dots,n$ and for $j=1,\dots,m$,
pairs $\langle (w_a^j,y_a^j),(\alpha^j,\beta^j)\rangle$, $\langle
(w_b^j,y_b^j),(\beta^j,\gamma^j)\rangle$, and $\langle
(w_c^j,y_c^j),(\gamma^j,\delta^j)\rangle$, for $j=1,\dots,m$, and the edges
incident to the dummy vertices $d^j_1,\dots,d^j_6$, for $j=1,\dots,m$. However,
it is easy to verify that $A$ assigns alternating edges to different sides of
$G_\cap$.

We show that $A$ satisfies the orthogonality constraints at every vertex. For all
the vertices except for $w_i^j$, $\beta^j$, $d_3^j$, $d_4^j$, and
$\gamma^j$, for $i=1,\dots,n$ and for $j=1,\dots,m$, this is true since they
have only one incident exclusive edge. For vertices $w_i^j$, $d_3^j$, and
$d_4^j$, with $i=1,\dots,n$ and for $j=1,\dots,m$, this is true since all the
edges incident to $w_i^j$ are assigned to the same side of $G_\cap$ by $A$, by
construction. 
For vertex $\beta_j$ and $\gamma_j$, we distinguish two cases based on whether
there exists a $\circ \in \{a,c\}$ with $T(x_b)=T(x_\circ)$:
(i) If this is case, let $\circ=c$ without loss of generality; the case
$\circ=a$ can be shown analogously.
Then, $A(\beta^j,\gamma^j)=A(\gamma^j,\delta^j)=A(d_4^j,\gamma^j)$, by construction, and hence the orthogonality constraints are satisfied at $\gamma^j$. To prove that they are also satisfied at $\beta^j$, it suffices to show that the two edges of $E_3$ incident to $\beta^j$ are assigned to different sides of $G_\cap$, given that $\beta^j$ has degree $3$ in $G_1$ and degree $2$ in $G_2$. 
Namely, due to the fact that $T$ is a \naesat truth assignment we have that $T(x_a) \neq T(x_b)$, and hence $A(\alpha^j,\beta^j) \neq A(\beta^j,\gamma^j)$.
(ii) In the second case, $T(x_a) \neq T(x_b) \neq T(x_c)$, hence we have that $A(\alpha^j,\beta^j) \neq A(\beta^j,\gamma^j) \neq A(\gamma^j,\delta^j)$. Since vertices $\beta^j$ and $\gamma^j$ have degree $4$ in $G_3$, degree $3$ in $G_1$, and degree $2$ in $G_2$, this implies that the orthogonality constraints are satisfied at $\beta^j$ and at $\gamma^j$.

Suppose that $\langle G_1=(V,E_1),G_2=(V,E_2),G_3=(V,E_3)\rangle$ is a positive
instance of \ourproblem{3} and let $A$ be the corresponding assignment of the
exclusive edges to the sides of $G_\cap$. We show how to construct a \naesat
truth assignment $T$ that satisfies $\phi$. For $i=1,\dots,n$, we set
$T(x_i)=\true$ if and only if $A(w_i^1,z_i^1)=l$.
We start by proving that, for each $i=1,\dots,n$, all the edges incident to
$w_i^j$, with $1 \leq j \leq m$, are assigned to the same side of $G_\cap$.
Observe that, for each $i=1,\dots,n$ and for each for each $j=1,\dots,m$, the
two edges in $G_3$ incident to $w_i^j$ both alternate with edge $(u_i^j,v_i^j)$
along $G_\cap$ and hence are assigned to the same side of $G_\cap$ by the
planarity constraints. Hence, by the orthogonality constraints at $w_i^j$, all the
exclusive edges in $E_1 \cup E_2$ incident to $w_i^j$ lie on the same side of
$G_\cap$ as $(w_i^j,z_i^j)$. Further, since any two vertices $w_i^j$ and
$w_i^{j+1}$, are connected by a transmission edge in either $E_1$
or in $E_2$, the statement follows.
This property allows us to focus on each clause separately. Let $c_j=(x_a^j,x_b^j,x_c^j)$ be a clause in $\phi$, with $1\leq j \leq m$. We show that $T(x_a^j)=T(x_b^j)=T(x_c^j)$ does not hold.
First, we show that $A(\beta^j,d_3^j) \neq A(d_4^j,\gamma^j)$. Namely, by the planarity constraints, $A(d_1^j,d_3^j)=A(d_3^j,d_5^j) \neq A(d_2^j,d_4^j) = A(d_4^j,d_6^j)$; then, by the orthogonality constraints at $d_3^j$ and at $d_4^j$, we have that $A(\beta^j,d_3^j)=A(d_1^j,d_3^j)=A(d_3^j,d_5^j)$ and that $A(d_4^j,\gamma^j) = A(d_2^j,d_4^j) = A(d_4^j,d_6^j)$, and the statement follows.
Second, $A(\alpha^j,\beta^j)=A(\beta^j,\gamma^j)=A(\gamma^j,\delta^j)$ does not hold, since $A(\beta^j,d_3^j) \neq A(d_4^j,\gamma^j)$ and by the orthogonality constraints at $\beta^j$ and at $\gamma^j$. This implies that $A(w_a^j,y_a^j)=A(w_b^j,y_b^j)=A(w_c^j,y_c^j)$ does not hold, and hence $A(w_a^j,z_a^j)=A(w_b^j,z_b^j)=A(w_c^j,z_c^j)$ does not hold, since all the edges incident to $w_i^j$, with $1 \leq j \leq m$, are assigned to the same side of $G_\cap$. This concludes the proof that $T(x_a^j)=T(x_b^j)=T(x_c^j)$ does not hold.

It is easy to see that the reduction can be performed in polynomial time and that it can be extended to any $k>3$ by subdividing two edges of $G_\cap$ for each additional graph $G_i$ and by introducing an exclusive edge between these vertices only belonging to $G_i$.
\end{proof}


\section{Omitted or Sketched Proofs from Section~\ref{se:cycle-algorithms}}\label{app:cycle-algorithms}

\rephrase{Lemma}{\ref{le:reduction-degree-three-2-degree-three-outerplanar}}{
Let \ourinstance{} be an instance of \ourproblem{2} such that $G_\cap = (V, E_1 \cap E_2)$ is a cycle and $G_1$ has maximum degree $3$. It is possible to construct in polynomial time an equivalent instance \ourinstance{*} of \ourproblem{2} such that $G_\cap^* = (V^*,E_1^* \cap E_2^*)$ is a cycle and $G_1^*$ is an outerplanar graph with maximum degree $3$.
}

\begin{proof}
We describe how to construct an equivalent instance \ourinstance{\prime} of \ourproblem{2} such that $G_\cap^\prime$ is a cycle, $G_1^\prime$ has maximum degree $3$ and the number of pairs of edges in $G_1^\prime$ that alternate along $G_\cap^\prime$ is smaller than the number of pairs of edges in $G_1$ that alternate along $G_\cap$.
Note that repeatedly performing this transformation eventually yields an equivalent instance \ourinstance{*} satisfying the requirements of the lemma.

Consider two edges $e=(u,v)$ and $f=(w,z)$ of $G_1$ such that $u,w,v,z$ appear in this order along cycle $G_\cap$ and such that the path $P_{u,z}$ in $G_\cap$ between $u$ and $z$ that contains $v$ and $w$ has minimal length. If $G_1$ is not outerplanar, edges $e$ and $f$ always exist.

\begin{figure}[tb]
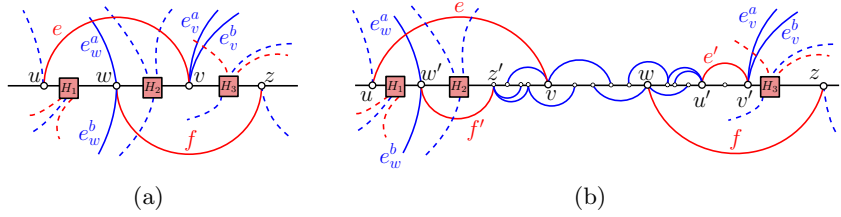

  \centering
  \subfigure[]{\includegraphics[scale=0.5,page=1]{img/red-deg3-outerplanar.pdf}\label{fig:reduction-deg3-2-outerplanar-a-app}}
  \hfil
  \subfigure[]{\includegraphics[scale=0.5,page=2]{img/red-deg3-outerplanar.pdf}\label{fig:reduction-deg3-2-outerplanar-b-app}}
  \caption[]{Instances (a) \ourinstance{} and (b) \ourinstance{\prime} for the
  proof of Lemma~\ref{le:reduction-degree-three-2-degree-three-outerplanar}.
  Edges of the shared graph $G_\cap$ are black. Exclusive edges of $G_1$ are red and those of $G_2$ are blue.}
  \label{fig:reduction-deg3-2-outerplanar-app}
\end{figure}

Initialize $G_\cap^\prime = G_\cap$. Replace path $P_{u,z}$ in $G_\cap^\prime$ by a path $P'_{u,z}$, as follows; refer to Fig.~\ref{fig:reduction-deg3-2-outerplanar-app}. 
Let $H_1$, $H_2$, and $H_3$ be the sets of vertices between $u$ and $w$, between $w$ and $v$, and between $v$ and $z$ in $G_\cap$. Path $P'_{u,z}$ contains $u$, then the vertices in $H_1$, then a dummy vertex $w'$, then the vertices in $H_2$, then a dummy vertex $z'$, then three dummy vertices $x_1$, $x_2$, $x_3$, then $v$, then four dummy vertices $x_4$, $x_5$, $x_6$, $x_7$, then $w$, then three dummy vertices $x_8$, $x_9$, $x_{10}$, then three dummy vertices $u'$, $x_{11}$, and $v'$, then the vertices in $H_3$, and finally $z$. Note that $G_\cap^\prime$ contains all the vertices of $G_\cap$, plus a set of dummy vertices. 
We now describe the exclusive edges in $E_1^\prime$ and $E_2^\prime$. 
Initialize $E_1^\prime=E_1$ and $E_2^\prime=E_2$.
Add edges $e'=(u',v')$ and $f'=(w',z')$ to $E_1^\prime$.
Also, add edges $(z',x_2)$, $(z',x_3)$, $(x_1,v)$, $(x_3,x_4)$, $(v,x_6)$, $(x_5,w)$, $(x_7,x_8)$, $(w,x_{10})$, $(x_8,u')$, and $(x_9,u')$ to $E_2^\prime$.
Finally, replace in $E_2^\prime$ each edge $(x,w)$ incident to $w$ by an edge $(x,w')$ and each edge $(x,v)$ incident to $v$ with an edge $(x,v')$.

Before proving the statement, we observe an important property that will be used in the following, namely that there exists no exclusive edge in $E_1$, and hence in $E_1^\prime$, with an endpoint in $H_2$ and the other one not in $H_2$. In fact, there exists no edge connecting a vertex of $H_2$ to any of $u,v,w,z$, since these vertices are already incident to edges $e$ and $f$, respectively, and since $G_1$ has maximum degree $3$. Also, there exists no exclusive edge $g$ connecting a vertex of $H_2$ to a vertex of $H_1$ (of $H_3$), since in this case $g$ would alternate with $f$ (with $e$), hence contradicting the minimality of path $P_{u,z}$. Finally, the existence of an exclusive edge connecting a vertex of $H_2$ to any other vertex in $V$ would immediately make the instance negative, since $G_1$ would not be planar.

We now prove that \ourinstance{\prime} satisfies the required properties. 
First, graph $G_\cap^\prime$ is a cycle by construction. 
Second, $G_1^\prime$ has maximum degree $3$, since (i) every vertex in $V \cap V^\prime$ is incident to the same edges in $E_1^\prime$ as in $E_1$, (ii) dummy vertices $x_i$, with $i=1,\dots,11$, have degree $2$, and (iii) dummy vertices $w'$, $z'$, $u'$, and $v'$ have degree $3$.
Third, the number of pairs of alternating edges in \ourinstance{\prime} is smaller than in \ourinstance{}. In fact, (i) edge $e'$ does not alternate with any edge of $E_1^\prime$, since $x_{11}$ is not incident to any exclusive edge in $E_1^\prime$, (ii) edge $f'$ does not alternate with any edge in $E_1^\prime$, since there exists no exclusive edge in $E_1^\prime$ with an endpoint in $H_2$ and the other one not in $H_2$, and (iii) all pairs of edges in $E_1 \cap E_1^\prime$ that alternate along $G_\cap^\prime$ also alternate along $G_\cap$, except for edges $e$ and $f$, which alternate along $G_\cap$ but not along $G_\cap^\prime$.

We now prove that \ourinstance{\prime} is equivalent to \ourinstance{}.

Suppose that \ourinstance{} admits an \ourdrawing \oursolution{}. By Theorem~\ref{th:characterization}, \oursolution{} determines an assignment $A$ of the exclusive edges of $E_1$ and of $E_2$ to the two sides of $G_\cap$ satisfying all the planarity and the orthogonality constraints. We show how to construct an assignment $A^\prime$ of the exclusive edges of $E_1^\prime$ and of $E_2^\prime$ to the two sides of $G_\cap^\prime$ satisfying all the constraints.

For each exclusive edge $g \in E_1 \cap E_1^\prime$, set $A^\prime(g) = A(g)$. Also, set $A^\prime(e') = A(e)$ and $A^\prime(f') = A(f)$.
For each exclusive edge $g \in E_2 \cap E_2^\prime$, set $A^\prime(g) = A(g)$. Also, for each edge $(x,w')$ (resp. $(x,v')$) incident to $w'$ (resp. to $v'$), set $A^\prime(x,w') = A(x,w)$ (resp. $A^\prime(x,v') = A(x,v)$). Further, set $A^\prime(x_1,v) = A^\prime(v,x_6) = A^\prime(x_7,x_8) = A^\prime(x_8,u') = A^\prime(x_9,u') = A(e)$ and set $A^\prime(z',x_2) = A^\prime(z',x_3) = A^\prime(x_3,x_4) = A^\prime(x_5,w) = A^\prime(w,x_{10}) = A(f)$.
 
The planarity constraints for the edges of $G^\prime_1$ are satisfied since any pair of edges that alternate along $G_\cap^\prime$ also alternate in $G_\cap$ and since their assignment in $A^\prime$ and in $A$ are the same, by construction.

We prove that the planarity constraints for the edges of $G^\prime_2$ are
satisfied by $A^\prime$. For the edges that are not incident to any dummy
vertex, this is true for the same reason as for the edges of $G^\prime_1$. 
For each edge $(x,w')$ incident to $w'$, this is true since $A^\prime(x,w')=A(x,w)$, and since $(x,w')$ alternates with an edge $g \in E_2^\prime$ along $G_\cap^\prime$ if and only if edge $(x,w)$ alternates with an edge $g^*$ along $G_\cap$, where $g^*=g$ if $g$ is not incident to $v'$, while $g^*=(y,v)$ if $g=(y,v')$. Analogous arguments hold for each edge $(x,v')$ incident to $v'$.
Finally, the fact that the planarity constraints for each edge incident to two dummy vertices are satisfied by $A^\prime$ can be easily verified; recall that $A(e) \neq A(f)$.

We now prove that the orthogonality constraints are satisfied by $A^\prime$ at every vertex of $V^\prime$. 
For the vertices in $V^\prime \cap V \setminus \{w,v\}$, this is true since they are satisfied by $A$ and since for every exclusive edge $g \in E_1^\prime \cup E_2^\prime$ incident to these vertices, we have that $g \in E_1 \cup E_2$, by construction, and $A^\prime(g) = A(g)$.
For vertex $w$, this is true since $A'(x_5, w) = A'(w,x_{10}) = A(f) = A'(f)$.
For vertex $v$, this is true since $A'(x_1, v) = A'(v,x_6) = A(e) = A'(e)$.
For vertex $u^\prime$, this is true since $A^\prime(x_8,u')=A^\prime(x_9,u')=A^\prime(e')=A(e)$.
For vertex $z^\prime$, this is true since $A^\prime(z',x_2) = A^\prime(z',x_3)=A^\prime(f')=A(f)$.
For vertex $w^\prime$, assume there exist two exclusive edges $e^a_w, e^b_w \in E_2^\prime$ that are incident to $w'$, the case in which there exists only one or none of them being trivial. Since $A^\prime(e^a_w)= A(e^a_w)$, $A^\prime(e^b_w)=A(e^b_w)$, and $A^\prime(f')=A(f)$, and since the orthogonality constraints at $w$ are satisfied by $A$, the orthogonality constraints at $w^\prime$ are satisfied by $A^\prime$.
Analogously, the orthogonality constraints at $v^\prime$ between edges $e^a_v, e^b_v \in E_2^\prime$, if any, and edge $e' \in E_1^\prime$ are satisfied by $A^\prime$ since the same constraints at $v$ between edges $e^a_v, e^b_v \in E_2$ and $e \in E_1$ are satisfied by $A$.
Since vertices $x_i$, with $i=1,\dots,11$, have degree $2$ in $G_1^\prime$, this concludes the proof that $A^\prime$ satisfies the orthogonality constraints.

Suppose that \ourinstance{\prime} admits \ourdrawing \oursolution{\prime}, and let $A^\prime$ be the corresponding assignment of the exclusive edges of $E_1^\prime$ and of $E_2^\prime$ to the two sides of $G_\cap^\prime$. We show how to construct an assignment $A$ of the exclusive edges of $E_1$ and of $E_2$ to the two sides of $G_\cap$ satisfying all the planarity and the orthogonality constraints. 

For each exclusive edge $g \in E_1$, set $A(g) = A^\prime(g)$. 
For each exclusive edge $g \in E_2 \cap E_2^\prime$, set $A(g) = A^\prime(g)$. Also, for each edge $(x,w)$ (resp. $(x,v)$) incident to $w$ (resp. to $v$), set $A(x,w) = A^\prime(x,w')$ (resp. $A(x,v) = A^\prime(x,v')$).

We prove that the planarity constraints for the edges of $G_1$ are satisfied by $A$.
For each pair $(e_1,e_2)$ of exclusive edges in $E_1$ such that $\{e_1, e_2\} \neq \{e, f\}$, this is true since $e_1$ and $e_2$ alternate along $G_\cap$ if and only if they alternate along $G^\prime_\cap$, by construction.
For pair $(e,f)$, this is true for the following reason. By planarity constraints, we have $A^\prime(x_1,v) = A^\prime(v,x_6) \neq A^\prime(x_3,x_4)$; hence, by orthogonality constraints at vertex $v$, we have $A^\prime(e)=A^\prime(x_1,v) = A^\prime(v,x_6)$. Analogously, we have $A^\prime(f)=A^\prime(x_5,w) = A^\prime(w,x_{10}) \neq A^\prime(x_7,x_8)$. Since, by planarity constraints, $A^\prime(v,x_6) \neq A^\prime(x_5,w)$, we have $A^\prime(e) \neq A^\prime(f)$ and hence $A(e) \neq A(f)$.
We prove that the planarity constraints for the edges of $G_2$ are satisfied by $A$. For each pair $(e_1,e_2)$ of exclusive edges in $E_2$ such that neither $e_1$ nor $e_2$ is incident to either of $w$ and $v$, this is true since $e_1$ and $e_2$ alternate along $G_\cap$ if and only if they alternate along $G^\prime_\cap$, by construction.
For each edge $(x,w)$ incident to $w$, this is true since $A(x,w)=A^\prime(x,w')$, and since $(x,w)$ alternates with an edge $g \in E_2$ along $G_\cap$ if and only if edge $(x,w)$ alternates with an edge $g^*$ along $G^\prime_\cap$, where $g^*=g$ if $g$ is not incident to $v$, while $g^*=(y,v')$ if $g=(y,v)$. Analogous arguments hold for each edge $(x,v)$ incident to $v$.

We now prove that the orthogonality constraints are satisfied by $A$ at every vertex of $V$.
For vertices in $V \setminus \{w,v\}$, this is true since they are satisfied by $A^\prime$ and since for every exclusive edge $g \in E_1 \cup E_2$ incident to these vertices, $g \in E_1^\prime \cup E_2^\prime$, by construction, and $A(g) =
A^\prime(g)$. 
In order to prove that the constraints are satisfied also at $w$ and $v$, we first argue that $A(f)=A^\prime(f')$ and $A(e)=A^\prime(e')$:
Namely, by planarity constraints, $A'(z',x_2) \neq A'(x_1, v) \neq A'(z',x_3)$, 
and hence $A'(z',x_2) = A'(z',x_3)$.
Similarly, $A'(x_5, w) \neq A'(x_7,x_8) \neq A'(w, x_{10})$, and hence $A'(x_5,w)
= A'(w, x_{10})$.
Then, by using the longer chain of alternating edges we get $A'(z', x_3) \neq
A'(x_1, v) \neq A'(x_3, x_4) \neq A'(v, x_6) \neq A'(x_5, w)$ and thus $A'(z',
x_3) = A'(x_5, w)$.
Finally, by orthogonality constraints at $z'$ and $w$, we get $A'(f') = A'(z', x_3)$ and $A'(f) = A'(x_5, w)$.
Since $A(f) = A'(f)$ we conclude $A(f) = A'(f')$.
The equality $A(e)=A^\prime(e')$ follows symmetrically.
We now prove that the orthogonality constraints are satisfied at $w$ and $v$.
For vertex $w$, assume there exist two exclusive edges $e^a_w, e^b_w \in E_2$ incident to $w$, the case in which there exists only one or none of them being trivial. Since $A(e^a_w)= A^\prime(e^a_w)$, $A(e^b_w)=A^\prime(e^b_w)$, and $A(f)=A^\prime(f')$, and since the orthogonality constraints at $w'$ between $e^a_w, e^b_w$ and $f'$ are satisfied by $A^\prime$, the orthogonality constraints at $w$ between $e^a_w, e^b_w$ and $f$ are satisfied by $A$.
Analogously, the orthogonality constraints at $v$ between edges $e^a_v, e^b_v \in E_2$, if any, and edge $e \in E_1$ are satisfied by $A$ since the same constraints at $v'$ between $e^a_v, e^b_v \in E^\prime_2$ and $e' \in E^\prime_1$ are satisfied by $A^\prime$.
This concludes the proof of the lemma.
\end{proof}

\rephrase{Lemma}{\ref{le:reduction-degree-five-2-degree-three}}{
Let \ourinstance{} be an instance of \ourproblem{2} such that $G_\cap = (V,E_1 \cap E_2)$ is a cycle and each vertex $v \in V$ has degree at most $3$ in either $G_1$ or $G_2$. It is possible to construct in polynomial time an equivalent instance \ourinstance{*} of \ourproblem{2} such that $G_\cap^* = (V^*,E_1^* \cap E_2^*)$ is a cycle and graph $G_1^*$ has maximum degree $3$.
}

\begin{proof}
We describe how to construct an equivalent instance \ourinstance{\prime} of \ourproblem{2} such that $G_\cap^\prime$ is a cycle, each vertex $v \in V^\prime$ has degree at most $3$ in either $G_1^\prime$ or $G_2^\prime$, and the number of degree-$4$ vertices in $G_1^\prime$ is smaller than the number of degree-$4$ vertices in $G_1$.
Note that repeatedly performing this transformation eventually yields an equivalent instance \ourinstance{*} satisfying the requirements of the lemma.

Consider a vertex $v \in V$ such that there exists two edges $e=(v,u_e), f=(v,u_f) \in E_1$ incident to $v$. Assume without loss of generality that $u_e$, $v$, and $u_f$ appear in this order along $G_\cap$. Suppose that there exists an edge $h=(v,u_h) \in E_2$ incident to $v$, the other case being simpler. We describe the construction for the case in which vertices
$u_e$, $v$, $u_f$, $u_h$ appear in this order along $G_\cap$; the other cases are analogous.

\begin{figure}[tb]
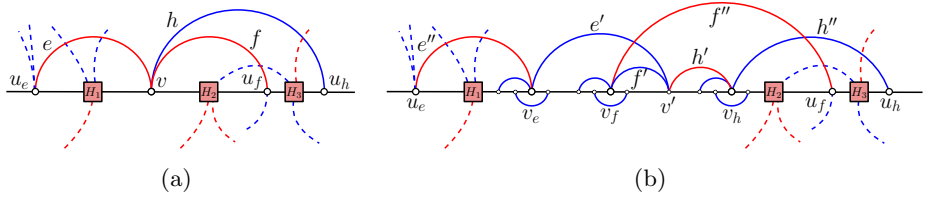

  \centering
  \subfigure[]{\includegraphics[scale=0.48,page=1]{img/red-deg5-to-degree3.pdf}\label{fig:red-deg5-to-degree3-a-app}}
  \hfill
  \subfigure[]{\includegraphics[scale=0.48,page=2]{img/red-deg5-to-degree3.pdf}\label{fig:red-deg5-to-degree3-b-app}}
  \caption[]{Instances (a) \ourinstance{} and (b) \ourinstance{\prime} for the
  proof of Lemma~\ref{le:reduction-degree-five-2-degree-three}.
  Edges of the shared graph $G_\cap$ are black. Exclusive edges of $G_1$ are red and those of $G_2$ are blue.}
  \label{fig:red-deg5-to-degree3-app}
\end{figure}

Initialize $G_\cap^\prime = G_\cap$; refer to Fig.~\ref{fig:red-deg5-to-degree3-app}. Replace $v$ in $G_\cap^\prime$ by a path\\ $P_{v}=x_1,x_2,v_e,x_3,y_1,y_2,v_f,y_3,v^\prime,z_1,z_2,v_h,z_3$ composed of dummy vertices.

We now describe the exclusive edges in $E_1^\prime$ and $E_2^\prime$. 
Set $E_i^\prime$, with $i=1,2$, contains all the exclusive edges in $E_i$ that are not incident to $v$. 
Also, $E_1^\prime$ contains edges $e''=(v_e,u_e)$, $f''=(v_f,u_f)$, and $h'=(v_h,v')$.
Finally, $E_2^\prime$ contains edges $(x_1,v_e)$, $(x_2,x_3)$, $(y_1,v_f)$, $(y_2,y_3)$, $(z_1,v_h)$, $(z_2,z_3)$, and edges $e'=(v_e,v')$, $f'=(v_f,v')$, and $h''=(v_h,u_h)$.

We prove that \ourinstance{\prime} satisfies the required properties. 
First, graph $G_\cap^\prime$ is a cycle by construction. 
Second, the degree of the vertices in $V \setminus V^\prime$ is the same in $G_1^\prime$ (resp. $G_2^\prime$) as in $G_1$ (resp. as in $G_2$), while all the dummy vertices have degree at most $3$ in $G_1^\prime$. Hence, every vertex in $V^\prime$ has degree at most $3$ in either $G_1^\prime$ or $G_2^\prime$; also, the number of degree-$4$ vertices in $G_1^\prime$ is smaller than the number of degree-$4$ vertices in $G_1$, since $v \notin G_1^\prime$.

We now prove that \ourinstance{\prime} is equivalent to \ourinstance{}.

Suppose that \ourinstance{} admits an \ourdrawing \oursolution{}, and let $A$ be the corresponding assignment of the exclusive edges of $E_1$ and of $E_2$ to the two sides of $G_\cap$, which exists by Theorem~\ref{th:characterization}.
We show how to construct an assignment $A^\prime$ of the exclusive edges of $E_1^\prime$ and of $E_2^\prime$ to the two sides of $G_\cap^\prime$ satisfying all the constraints.

For each exclusive edge $g \in E_1 \cup E_2$ incident to $v$, set $A^\prime(g) = A(g)$. Also, set $A^\prime(e'') = A(e)$, $A^\prime(f'') = A(f)$, and $A^\prime(h') = A(h)$.
Finally, set $A^\prime(x_1,v_e) = A^\prime(e') = A(e)$ and 
$A^\prime(x_2,x_3) \neq A(e)$;
set $A^\prime(y_1,v_f) = A^\prime(f') = A(f)$ and 
$A^\prime(y_2,y_3) \neq A(f)$; and
set $A^\prime(z_1,v_h) = A^\prime(h'') = A(h)$ and  
$A^\prime(z_2,z_3) \neq A(h)$.

We prove that the planarity constraints for the edges of $G^\prime_1$ are satisfied by $A^\prime$. Note that, by construction, 
edge $h'$ does not alternate with any edge of $G^\prime_1$ along $G_\cap^\prime$. Also, edges $e''$ and $f''$ do not alternate with each other along $G_\cap^\prime$. Further, if edge $e''$ (edge $f''$) alternates with an edge $g \in G^\prime_1$ along $G_\cap^\prime$, then edge $e$ (edge $f$) alternates with $g$ along $G_\cap$. Finally, any two edges not incident to any dummy vertex that alternate along $G_\cap^\prime$ also alternate along $G_\cap$. In all the described cases, the planarity constraints are satisfied by $A^\prime$ since they are satisfied by $A$.

We prove that the planarity constraints for the edges of $G^\prime_2$ are satisfied by $A^\prime$.  Note that, by construction, 
edges $e'$, $f'$, $(x_1,v_e)$, $(x_2,x_3)$, $(y_1,v_f)$, $(y_2,y_3)$, $(z_1, v_h)$, and $(z_2,z_3)$ do not alternate with any edge of $G^\prime_2$ that is not incident to a dummy vertex along $G_\cap^\prime$; it easy to verify that $A^\prime$ satisfies the planarity constraints among these edges.
Also, edge $h''$ alternates with $(z_2,z_3)$, but $A^\prime(h'') \neq A^\prime(z_2,z_3)$ by construction. Further, if edge $h''$ alternates with an edge $g \neq (z_2,z_3) \in G^\prime_2$ along $G_\cap^\prime$, then edge $h$ alternates with $g$ along $G_\cap$. Finally, any two edges not incident to any dummy vertex that alternate along $G_\cap^\prime$ also alternate along $G_\cap$. In all these cases, the planarity constraints are satisfied by $A^\prime$ since they are satisfied by $A$.

We now prove that the orthogonality constraints are satisfied by $A^\prime$ at every vertex in $V'$. 
For non-dummy vertices, this is true since they are satisfied by $A$ and since all the edges incident to them have the same assignment in $A$ as in $A^\prime$. 
For vertices $x_i$, $y_i$, and $z_i$, with $i=1,2,3$, this is true since they have degree $2$ in $G'_1$.
For vertex $v_e$, this is true since $A'(e'')=A'(x_1,v_e)=A'(e')= A(e)$; similar arguments apply for vertices $v_f$ and $v_h$.
Finally, for vertex $v'$, this is true since (i) $A'(e')=A(e)$, $A'(f')=A(f)$, and $A'(h')=A(h)$, (ii) $e$, $f$, and $h$ are incident to $v$ in $G_2$, and (iii) $A$ satisfies the orthogonality constraints. This concludes the proof that $A^\prime$ satisfies the orthogonality constraints.

Suppose that \ourinstance{\prime} admits an \ourdrawing \oursolution{\prime}, and let $A^\prime$ be the corresponding assignment of the exclusive edges of $E_1^\prime$ and of $E_2^\prime$ to the two sides of $G_\cap^\prime$. We show how to construct an assignment $A$ of the exclusive edges of $E_1$ and of $E_2$ to the two sides of $G_\cap$ satisfying all the planarity and the orthogonality constraints. 

For each exclusive edge $g \in E_1 \cup E_2$ not incident to $v$, set $A(g) = A^\prime(g)$. Also, set $A(e) = A^\prime(e'')$, $A(f) = A^\prime(f'')$, and $A(h) = A^\prime(h'')$.

We prove that the planarity constraints for the edges of $G_1$ and of $G_2$ are satisfied by $A$.
Consider any pair of edges $\langle g_1,g_2\rangle$ of the same graph $G_i$, with $i=1,2$, that alternate along $G_\cap$. If none of $g_1$ and $g_2$ is incident to $v$, then they also alternate along $G'_\cap$. Hence, the planarity constraints are satisfied by $A$ since they are satisfied by $A'$.
Otherwise, assume $g_1$ is incident to $v$; note that $g_2$ is not incident to $v$, since $g_1$ and $g_2$ alternate along $G_\cap$. If $g_1 = e$ (if $g_1 = f$; if $g_1 = h$), then edge $e''$ (edge $f''$; edge $h''$) alternates with $g_2$ along $G'_\cap$. Further, $A(e)=A'(e'')$, $A(f)=A'(f'')$, $A(h)=A'(h'')$, and $A(g_2)=A'(g_2)$. Hence, the planarity constraints for these edges are satisfied by $A$ since they are satisfied by $A'$.

We finally prove that the orthogonality constraints are satisfied by $A$.
For the vertices in $V \setminus \{v,u_e,u_f,u_h\}$, this is true since they are satisfied by $A^\prime$ and since for every exclusive edge $g \in E_1 \cup E_2$ incident to these vertices, we have that $g \in E_1^\prime \cup E_2^\prime$, by construction, and $A(g) = A^\prime(g)$. 
For vertex $u_e$, this is true since for each edge $g$ incident to $u_e$ different from $e$, it holds that $A(g)=A'(g)$, since $A(e)=A'(e'')$, and since the orthogonality constraints at $u_e$ are satisfied by $A'$. Analogous arguments hold for vertices $u_f$ and $u_h$.
To prove that this is true for $v$, we first argue that $A(e)=A^\prime(e')$, that $A(f)=A^\prime(f')$, and that $A(h)=A^\prime(h')$:
Namely, by planarity constraints, we get $A'(e') = A'(x_1,v_e)$ since they both alternate with $(x_2,x_3)$; hence, by orthogonality constraints at $v_e$, we get $A'(e'') = A'(e')$. Since $A(e) = A'(e'')$, by construction, we conclude $A(e) = A'(e')$. The equalities $A(f)=A^\prime(f')$ and $A(h)=A^\prime(h')$ follow symmetrically.
Hence, the orthogonality constraints at $v$ are satisfied by $A$ since they are satisfied at $v'$ by $A^\prime$. 
This concludes the proof.
\end{proof}

\rephrase{Theorem}{\ref{th:algorithm-one-degree-three}}{
\ourproblem{2} can be solved in polynomial time for instances 
whose shared graph is a cycle and whose union graph has maximum degree~$5$.
}

\begin{proof}
First apply Lemma~\ref{le:reduction-degree-five-2-degree-three} to obtain an equivalent instance \ourinstance{\prime} such that $G_\cap^\prime$ is a cycle and graph $G_1^\prime$ has maximum degree $3$. Then, apply Lemma~\ref{le:reduction-degree-three-2-degree-three-outerplanar} to obtain an equivalent instance \ourinstance{\prime\prime} such that $G_\cap^{\prime\prime}$ is a cycle and  $G_1^{\prime\prime}$ is an outerplanar graph with maximum degree $3$. Finally, apply Lemma~\ref{le:algorithm-degree-three-outerplanar} to test in polynomial time whether \ourinstance{\prime\prime}, and hence \ourinstance{}, is a positive instance.
\end{proof}

\section{Omitted or Sketched Proofs from Section~\ref{se:biconnected}}\label{app:biconnected}

\rephrase{Lemma}{\ref{le:simple-attachments}}{
  Let \ourinstance{} be an instance of \ourproblem{2} whose shared graph 
  is biconnected. It is possible to construct in polynomial time an equivalent instance \ourinstance{*}
  whose shared graph is biconnected and such that each endpoint of an
  exclusive edge has degree~$2$ in the shared graph.
  }

\begin{proof}
  We start with a simplification step that removes certain edges.  An
  exclusive edge $e=uv$ of $G_1$ or of $G_2$ is an \emph{intra-pole}
  edge if its endpoints are adjacent in some skeleton of the SPQR-tree
  of the shared graph $G_\cap$.  If neither $u$ nor $v$ is
  incident to other exclusive edges, then $uv$ is \emph{isolated}.
  Let $E_1'$ and $E_2'$ be the isolated intra-pole edges of $G_1$ and
  $G_2$, respectively.  
  
  We claim that the instance $\langle G_1 -
  E_1', G_2 - E_2' \rangle$ admits an \ourdrawing if and only if
  \ourinstance{} does.
  The if part is clear since we can simply remove the isolated
  intra-pole edges from an \ourdrawing of \ourinstance{} to obtain
  an \ourdrawing.  Conversely, Angelini et al.~\cite{adfpr-tsetgibgt-11} show that the intra-pole edges can be
  reinserted into any {\sc SEFE}, and thus also into an \ourdrawing of
  $\langle G_1 -E_1', G_2 -E_2' \rangle$ without crossings, i.e., the
  planarity constraints are satisfied for \ourinstance{}.  Since the
  edges in $E_1' \cup E_2'$ are isolated also the orthogonal
  constraints are trivially satisfied, and we obtain an \ourdrawing of \ourinstance{}.  This finishes the proof of the
  claim.

  In the following, we assume that \ourinstance{} has been
  preprocessed in this way, and it hence does not contain isolated
  intra-pole edges.

  \begin{figure}[t]
  \centering
  \includegraphics{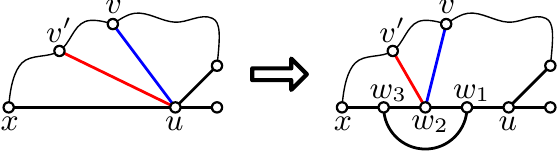}
  \caption{Moving exclusive edges from a vertex with degree~$3$ in the shared graph to a new vertex with degree~$2$ in the shared graph.}
  \label{fig:biconnected-replacement}
\end{figure}

  Consider an exclusive edge $e=uv$ in $G_1$ or $G_2$, say in $G_1$,
  such that $u$ has degree~3 in the shared graph.  Assume that $ux$ is
  an edge of $G$ incident to $u$ such that, in every \ourdrawing
  of \ourinstance{} the edge $uv$ is embedded in a face of $G$
  incident to $ux$ (we describe how to determine such an edge later).  
  We perform the following transformation. We
  subdivide $ux$ by three vertices $w_1,w_2,w_3$ and add the edge
  $w_1w_3$.  We further replace $uv$ by $w_2v$ and also, if it exists,
  the (unique) exclusive edge $e' = uv'$ (from $G_2$) by $w_2v'$; see 
  Fig.~\ref{fig:biconnected-replacement}. 
  Call the resulting instance \ourinstance{'}.  It is not difficult to
  see that \ourinstance{} admits an \ourdrawing if and only if
  \ourinstance{'} does.  If \ourinstance{'} admits an \ourdrawing,
  we can contract the vertices $w_1,w_2,w_3$ onto $u$ to obtain an
  \ourdrawing of \ourinstance{}.  Note that the orthogonal
  constraint at $u$ is satisfied since the triangle $w_1,w_2,w_3$
  ensures that the exclusive edges incident to $u$ are embedded in the
  same face of $G_\cap'$ and hence in $G_\cap$.  Conversely, given an \ourdrawing of \ourinstance{}, due to the orthogonality constraints at
  $u$ all exclusive edges incident to $u$ are embedded in the same
  face of $G_\cap$, and hence the replacement can be carried out locally
  without creating crossings. Note that, after the transformation,
  there are fewer endpoints of exclusive edges that have degree~$3$ in
  the shared graph.  We iteratively apply this transformation to
  obtain the instance \ourinstance{*}.

  It remains to show that there always exists a suitable edge $ux$.
  Let $\mathcal T$ denote the SPQR-tree of the shared graph $G_\cap$.  
  Since $u$ has degree~3, there is exactly one node $\mu$
  of $\mathcal T$ whose skeleton contains $u$ and where the degree of
  $u$ in $\skel(\mu)$ is~3.  Note that $\mu$ is either a P-node or an
  R-node.

  First assume $\mu$ is an R-node and consider the position of $v$
  inside $\skel(\mu)$, where it is either a vertex of $\skel(\mu)$ or
  it is contained in a virtual edge $\eps_v$ of $\skel(\mu)$.  Since
  $\mu$ is an R-node, $u$ and $v$ ($u$ and $\eps_v$) share at most two
  faces, both of which are incident to a virtual edge $\eps$ incident
  to $u$.  We choose $ux$ as the (unique) edge incident to $u$ that is
  contained in the subgraph represented by $\eps$.

  Second, assume $\mu$ is a P-node.  If the other endpoint $v$ of $e$
  is not a pole of $\mu$, then $v$ is contained in a virtual edge
  $\eps_v$ of $\skel(\mu)$, and we can proceed as in the previous
  case; see Fig.~\ref{fig:biconnected-replacement}. Now assume that $v$ is the other vertex of
  $\skel(\mu)$, i.e., $e$ is an intra-pole edge.  Since $e$ cannot be
  isolated (due to the simplification step at the beginning), there
  exists an exclusive edge $e'$ in $G_2$ incident to $u$ or $v$.
  Since $e \ne e'$, the edge $e'$ has an endpoint $v'$ that is
  contained in a subgraph represented by a virtual edge $\eps$ of
  $\skel(\mu)$.  It follows that in every planar embedding of $G_2$,
  the edge $e'$ is embedded in a face incident to $\eps$.  By the
  orthogonality constraints at the vertex shared by $e$ and $e'$, $e$ also
  has to be embedded in a face incident to $\eps$ in any \ourdrawing. 
  We thus choose $ux$ as the (unique) edge incident to
  $u$ contained in the subgraph represented by $\eps$.
\end{proof}

\rephrase{Lemma}{\ref{le:s-nodes-independent}}{
  Let \ourinstance{} be an instance of \ourproblem{2} such that
    the shared graph $G_\cap$ is biconnected.  Then
    \ourinstance{} admits an \ourdrawing if and only if all
    instances \ourinstance{\mu} admit an \ourdrawing.
}

\begin{proof}
  It is not hard to see that each \ourinstance{\mu} can be obtained
  from \ourinstance{} by removing some vertices and edges and
  suppressing subdivision vertices.  Thus, if \ourinstance{} admits
  an \ourdrawing, so does each \ourinstance{\mu}.

  Conversely, assume that each \ourinstance{\mu} admits an \ourdrawing.  Recall that we have fixed a reference embedding
  for each skeleton of the SPQR-tree of the shared graph $G_\cap$ up to
  a flip.  
  We fix the flips of all reference embeddings
  as follows.  For each S-node $\mu$ and each neighbor $\nu$,
  represented by a virtual edge $\eps$ in $\skel(\mu)$, we consider
  the flips of the cycle $C_\eps$ in the \ourdrawing of
  \ourinstance{\mu} with respect to the ordering $O_\eps$ of the attachments of the subgraph
  represented by $\eps$.  If the reference embedding is used, we label the
  edge $\mu\nu$ with label~$1$, otherwise we label it $-1$.  Finally,
  we choose an arbitrary root $\mu_0$ of the augmented SPQR-tree for
  which we fix the reference embedding.  For each skeleton
  $\skel(\mu)$, $\mu \ne \mu_0$, we choose the reference embedding if
  and only if the product of the labels on the (unique) path from
  $\mu_0$ to $\mu$ is $1$, and its flip otherwise.  We denote the
  planar embedding of $G_\cap$ obtained in this way by $\mathcal E$.  
  
  It remains to determine the embeddings of $G_1$ and $G_2$.  After
  suitably flipping the given \ourdrawings, we can assume that
  their embeddings can be obtained from $\mathcal E$ by removing
  vertices and edges, and by contracting edges.  We now determine the
  embeddings of $G_1$ and $G_2$ as follows.  Recall that every vertex
  that is incident to exclusive edges has degree~2 in the shared
  graph.  For each vertex $v$ that is incident to exclusive edges of
  $G_1$ (of $G_2$), we consider the unique S-node $\mu$ whose skeleton
  contains $v$, and we choose the edge ordering as in the given \ourdrawing of \ourinstance{\mu}.  We claim that this results in an
  \ourdrawing $\langle \mathcal E_1, \mathcal E_2 \rangle$ of
  \ourinstance{}. Refer to Figs.~\ref{fig:biconnected-reduction-a} and~\ref{fig:biconnected-reduction-b}. 

  First observe that the orthogonality constraints are satisfied, 
  since the edge ordering of each vertex is chosen according to one of the given
  \ourdrawings.  It remains to show that the embeddings also
  satisfy the planarity constraints.  Due to the construction of the
  embeddings, all the exclusive edges are embedded in faces of
  $\mathcal E$; otherwise we would observe crossings in the skeletons
  of the (augmented) SPQR-tree.  Consider two exclusive edges $uv$ and
  $u'v'$ from the same graph that cross.  Since $uv$ and $u'v'$ cross,
  there exists a node $\mu$ of the (augmented) SPQR-tree such that for
  each of the two edges the endpoints are in different parts of
  $\skel(\mu)$.  If $\mu$ is a P-node or an R-node and all four parts
  containing these endpoints are distinct, then the parts containing
  the endpoints of these edges alternate around a face of
  $\skel(\mu)$.  This contradicts the planarity of the corresponding
  input graph $G_1$ or $G_2$.  Thus, in this case at least two
  attachments are contained in the same virtual edge $\eps$ of
  $\skel(\mu)$.  Let $\nu$ be the S-node of the augmented SPQR-tree
  corresponding to $\eps$.  Clearly, in $\skel(\nu)$, the endpoints
  of each of the two edges are distinct parts of $\skel(\nu)$.  It
  follows that the endpoints of the two edges alternate around the two
  faces of \ourinstance{\nu} corresponding to the two faces of $\skel(\mu)$. By
  construction of $\langle \mathcal E_1, \mathcal E_2 \rangle$ this
  contradicts the assumption that the given drawing of
  \ourinstance{\nu} is an \ourdrawing.
\end{proof}

\rephrase{Lemma}{\ref{le:cycle-to-path-gadget}}{
  \ourinstance{\mu} admits an \ourdrawing if and only if
  $\langle \overline{G_1^{\mu}}, \overline{G_2^{\mu}}\rangle$ does.}

\begin{proof}
  We simply show that, in terms of embeddings, the path $P_\eps$
  replacing~$C_\eps$ behaves the same as $C_\eps$.  First, observe that the
  edge $(a_2,x_3)$ ensures that all exclusive edges of $G_1$ incident to
  the clockwise $uv$-path of $C_i$ are embedded on the same side of
  the path $P_\eps$.  Similarly, $(x_2,b_1)$ ensures that all exclusive
  edges of $G_1$ incident to the counterclockwise $uv$-path of $C_\eps$
  are embedded on the same side of the path $P_\eps$.  Moreover, since
  the endpoints of the edges $(a_2,x_3)$ and $(x_2,b_1)$ alternate along
  $P_\eps$, they are embedded on different sides of $P_\eps$. Thus, 
  the exclusive edges of~$G_1$ incident to the clockwise and
  counterclockwise $uv$-path of $C_\eps$ cannot be embedded on the same
  side of $P_\eps$.  Similarly, the exclusive edges $(a_1,x_4)$ and $(x_2,b_2)$
  ensure that the exclusive edges of $G_2$ incident to the clockwise
  $uv$-path are all on one side of $P_\eps$ and the exclusive edges of
  $G_2$ incident to the counterclockwise $uv$-path are on the other
  side of $P_\eps$.  Finally, due to the alternation with $(a_2,x_3)$, the
  edges $(x_2,x_4)$ and $(x_2,b_1)$ must be embedded on the same side of
  $P_\eps$.  By the orthogonality constraint at $x_2$, the edge $(x_2,b_2)$ 
  must be also embedded on the same side as $(x_2,x_4)$.  Thus, $(a_2,x_3)$ and
  $(a_1,x_3)$ are on the same side of $P_\eps$ and likewise for~$(x_2,b_1)$ and
  $(x_2,b_2)$.  This ensures that the exclusive edges of $G_1$ and $G_2$
  incident to the clockwise $uv$-path of $C_\eps$ are embedded on the
  same side of $C_\eps$ and likewise for those incident to the
  counterclockwise $uv$-path.
\end{proof}

\rephrase{Theorem}{\ref{th:biconnected geometric}}{
  Let \ourinstancek{} be a positive instance of \ourproblem{k} whose shared graph is biconnected. Then,
  there exists an \ourdrawing \oursolutionk{} of \ourinstancek{} in
  which every edge has at most three bends.
}

\begin{proof}
  We assume that a cyclic order of the edges of the union graph around
  each vertex is given such that (a) it induces a planar embedding on
  each $G_i$, $i=1,\dots,k$, and (b) we can assign the incident
  edges around a vertex to at most four ports such that at most one edge
  of each $G_i$ is assigned to the same port.

  We adopt the method of Biedl and Kant~\cite{biedl/kant:98}. First,
  we compute in linear time~\cite{brandes:02} an $s$-$t$-ordering on
  the shared graph, i.e., we label the vertices $v_1,\dots,v_n$ such
  that $\{v_1,v_n\}$ is an edge of the shared graph and, for each
  $i=2,\dots,n-1$, there are $j < i < k$ such that $\{v_j,v_i\}$ and
  $\{v_i,v_k\}$ are edges of the shared graph. We choose the face to
  the left of $(v_1,v_n)$ as the outer face of the union graph.

  We now draw the union graph by adding the vertices in the order in
  which they appear in the $s$-$t$-ordering while respecting the given 
  order of the edges around each vertex. The edges will bend
  at most on y-coordinates near their incident vertices and are drawn
  vertically otherwise. We draw the edges around~$v_1$ as indicated in
  Fig.~\ref{FIG:geo_v1} where some of the incident edges might
  actually indicate several exclusive edges~--~at most one from each
  graph.

  \begin{figure}[tb!]
\centering
\hfil
  \subfigure[\label{FIG:geo_v1}around $v_1$]{\includegraphics[page=1]{img/geo.pdf}}
  \hfil
  \subfigure[\label{FIG:geo_vi}around $v_2,\dots,v_{n-1}$]{\includegraphics[page=4]{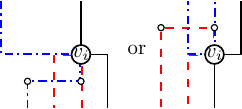}}
  \hfil
  \subfigure[\label{FIG:geo_vn}around $v_n$]{\includegraphics[page=5]{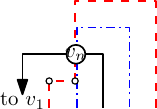}}

\caption{\label{FIG:geo}Constructing a drawing with at most three bends per edge}
\end{figure}

  For $i=2,\dots,n-1$, an edge may only leave $v_i$ to the bottom
  if it is incident to a neighbor with a lower index. 
  Again, some of the ports might host several exclusive edges, even one
  to a vertex with a lower index and one to a vertex with a higher
  index. Special cases occur when the ordering around $v_i$ is such that
  four exclusive edges of two distinct graphs must be assigned to two
  consecutive ports.  In particular, an edge leaving $v_i$ to a vertex with a
  lower index might bend twice around $v_i$ (see, e.g., the two small circles
  in Fig.~\ref{FIG:geo_vi}).

  Finally, the edges around $v_n$ are placed such that the edge
  $\{v_1,v_n\}$ enters it from the left. Thus, there are exactly three
  bends on $\{v_1,v_n\}$; see Fig.~\ref{FIG:geo_vn}. For any other
  edge, there is at most one bend around the endvertex with lower
  index and at most two bends around the endvertex with higher index.
\end{proof}

\section{Omitted or Sketched Proofs from Section~\ref{se:conclusions}}\label{app:conclusions}

\rephrase{Lemma}{\ref{le:reduction-general-2-degree-3,5}}{
Let \ourinstance{} be an instance of \ourproblem{2} whose shared graph $G_\cap$ is a cycle.
It is possible to construct in polynomial time an equivalent instance \ourinstance{*} of \ourproblem{2} such that 
(i) the shared graph $G_\cap^*$ is a cycle,
(ii) graph $G_1^*$ is outerplanar, and
(iii) no two degree-$4$ vertices in $G_1^*$ are adjacent to each other.
}

\begin{proof}
The reduction works in two steps. In the first step, we construct an instance \ourinstance{+} satisfying properties (i) and (iii) that is equivalent to \ourinstance{}; then, in the second step we construct the final instance \ourinstance{*} equivalent to \ourinstance{+}, which also satisfies property (ii).

For the first step, we show how to construct an instance \ourinstance{\prime} of \ourproblem{2} equivalent to \ourinstance{} such that $G_\cap^\prime$ is a cycle and the number of vertices with degree $4$ in $G_1^\prime$ not satisfying the condition of property (iii) is smaller than the number of vertices with degree $4$ in $G_1$ not satisfying this condition. Repeatedly performing this transformation eventually yields the required instance \ourinstance{+}.

Consider a vertex $v$ with degree $4$ in $G_1$ not satisfying the condition of property (iii). Let $e=(u,v)$ and $f=(v,w)$ be the two exclusive edges of $G_1$ incident to $v$. Assume that $u$, $v$, and $w$ appear in this order along $G_\cap$, the other cases being analogous.

Initialize $G_\cap^\prime = G_\cap$; refer to Fig.~\ref{fig:reduction-general-deg3point5-app}. Replace $v$ in $G_\cap^\prime$ by a path $P_{v}$ composed of dummy vertices $x_1, x_2, v_a, x_3,\dots,x_8, u', x_9, x_{10}$, of vertex $v$, and of dummy vertices $y_1,y_2,w',y_3,\dots,y_8,v_b,y_9,y_{10}$. Note that $G_\cap^\prime$ contains all the vertices of $G_\cap$, plus a set of dummy vertices. 
	\begin{figure}[tb]
		\centering
		\subfigure[]{\includegraphics[height=2cm,page=1]{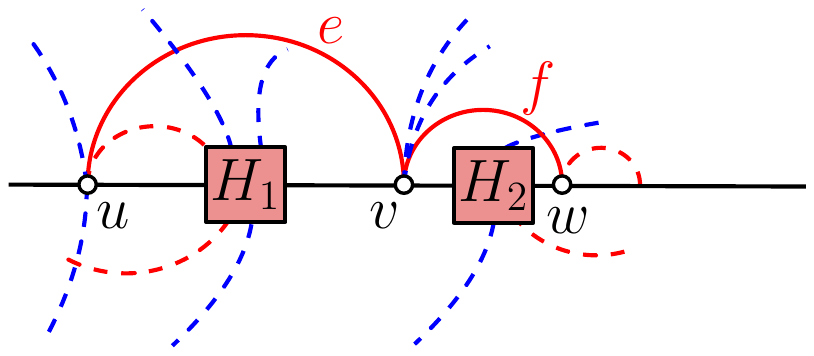}\label{fig:reduction-general-deg3point5-a-app}}
    \hfill
		\subfigure[]{\includegraphics[height=2cm,page=2]{img/red-general-deg3point5.pdf}\label{fig:reduction-general-deg3point5-b-app}}
		\caption{Illustrations for the proof of Lemma~\ref{le:reduction-general-2-degree-3,5}}\label{fig:reduction-general-deg3point5-app}
	\end{figure}
We now describe the exclusive edges in $E_1^\prime$ and $E_2^\prime$. 
Set $E_i^\prime$, with $i=1,2$, contains all the exclusive edges in $E_i$, except for $e$ and $f$. Also, $E_1^\prime$ contains edges $e'=(u,v_a)$, $e''=(u',v)$, $f''=(v,w')$, and $f'=(v_b,w)$. Finally, $E_2^\prime$ contains edges $(x_1,v_a)$, $(x_2,x_3)$, $(v_a,x_5)$, $(x_4,x_7)$, $(x_6,u')$, $(x_8,x_9)$, $(u',x_{10})$, $(y_1,w')$, $(y_2,y_3)$, $(w',y_5)$, $(y_4,y_7)$, $(y_6,v_b)$, $(y_8,y_9)$, and $(v_b,y_{10})$.

We prove that \ourinstance{\prime} satisfies the required properties. 
First, graph $G_\cap^\prime$ is a cycle by construction. 
Second, the number of vertices of degree $4$ in $G_1^\prime$ not satisfying the condition of property (iii) is smaller than the number of such vertices in $G_1$. In fact, any vertex $x \neq v$ with degree $4$ in $G_1^\prime$ satisfies the required condition if and only if it satisfies the same condition in $G_1$. On the other hand, vertex $v$ does not satisfy the condition in $G_1$, by hypothesis, but it satisfies the condition in $G_1^\prime$, since $u'$ and $w'$ have degree $3$ in $G_1^\prime$ and the path between them along $G_\cap^\prime$ containing $v$ only contains dummy vertices $x_9$, $x_{10}$, $y_1$, and $y_2$, which are not incident to any exclusive edge of $G_1^\prime$, by construction.

We now prove that \ourinstance{\prime} is equivalent to \ourinstance{}.

Suppose that \ourinstance{} admits an \ourdrawing \oursolution{}, and let $A$ be the corresponding assignment of the exclusive edges of $E_1$ and of $E_2$ to the two sides of $G_\cap$, which exists by Theorem~\ref{th:characterization}.
We show how to construct an assignment $A^\prime$ of the exclusive edges of $E_1^\prime$ and of $E_2^\prime$ to the two sides of $G_\cap^\prime$ satisfying all the constraints.

For each exclusive edge $g \in E_1 \cap E_1^\prime$, set $A^\prime(g) = A(g)$. Also, set $A^\prime(e') = A^\prime(e'') = A(e)$ and $A^\prime(f') = A^\prime(f'') = A(f)$.
For each exclusive edge $g \in E_2 \cap E_2^\prime$, set $A^\prime(g) = A(g)$. Also, set $A^\prime(x_1,v_a) = A^\prime(v_a,x_5) = A^\prime(x_6,u') = A^\prime(u',x_{10}) = A^\prime(y_1,w') = A^\prime(w',y_5) = A^\prime(y_6,v_b) = A^\prime(v_b,y_{10}) = A(e)$ and set $A^\prime(x_2,x_3) = A^\prime(x_4,x_7) = A^\prime(x_8,x_9) = A^\prime(y_2,y_3) = A^\prime(y_4,y_7) = A^\prime(y_8,y_9) = A(f)$.

We prove that the planarity constraints for the edges of $G^\prime_1$ are satisfied by $A^\prime$. Note that, by construction, edges $e''$ and $f''$ do not alternate with any edge of $G^\prime_1$ along $G_\cap^\prime$. Also, edges $e'$ and $f'$ do not alternate with each other. Further, if edge $e'$ (edge $f'$) alternates with an edge $g \in G^\prime_1$ along $G_\cap^\prime$, then edge $e$ (edge $f$) alternates with $g$ along $G_\cap$. Finally, any two edges different from $e'$, $e''$, $f'$, $f''$ that alternate along $G_\cap^\prime$ also alternate along $G_\cap$. In all the described cases, the planarity constraints are satisfied by $A^\prime$ since they are satisfied by $A$.

We prove that the planarity constraints for the edges of $G^\prime_2$ are
satisfied by $A^\prime$. Note that, by construction, edges in $E_2^\prime \cap E_2$ do not alternate with any edge incident to a dummy vertex along $G_\cap^\prime$, and alternate with each other along $G_\cap^\prime$ if only if they alternate with each other along $G_\cap$. Hence, the planarity constraints for these edges are satisfied by $A^\prime$ since they are satisfied by $A$. On the other hand, it is easily verified that the planarity constraints are satisfied by $A^\prime$ also for the edges incident to dummy vertices.

We now prove that the orthogonality constraints are satisfied by $A^\prime$ at every vertex in $V^\prime$. 
For the non-dummy vertices in $V^\prime \setminus \{u,v,w\}$, this is true since they are satisfied by $A$ and since the edges incident to these vertices have the same assignment in $A$ as in $A^\prime$.
For vertex $u$, this is true since they are satisfied by $A$, since $A^\prime(e')=A(e)$, and since the other edges incident to $u$ have the same assignment in $A$ and in $A^\prime$.
Analogously, for $w$ this is true since they are satisfied by $A$, since $A^\prime(f')=A(f)$, and since the other edges have the same assignment in $A$ and in $A^\prime$.
For $v$, this is true since they are satisfied by $A$, since $A^\prime(e'')=A(e)$, since $A^\prime(f'')=A(f)$, and since the other edges have the same assignment in $A$ and in $A^\prime$.
For $v_a$, this is true since $A^\prime(x_1,v_a)=A^\prime(v_a,x_5)=A^\prime(e')=A(e)$.
For $u'$, this is true since $A^\prime(x_6,u')=A^\prime(u',x_{10})=A^\prime(e'')=A(e)$.
For $w'$, this is true since $A^\prime(y_1,w')=A^\prime(w',y_5)=A^\prime(f'')=A(f)$.
For $v_b$, this is true since $A^\prime(y_6,v_b)=A^\prime(v_b,y_{10})=A^\prime(f')=A(f)$. Since all the other dummy vertices have degree $2$ in $G_1'$, this concludes the proof that $A^\prime$ satisfies the orthogonality constraints.

Suppose that \ourinstance{\prime} admits \ourdrawing \oursolution{\prime}, and let $A^\prime$ be the corresponding assignment of the exclusive edges of $E_1^\prime$ and of $E_2^\prime$ to the two sides of $G_\cap^\prime$. We show how to construct an assignment $A$ of the exclusive edges of $E_1$ and of $E_2$ to the two sides of $G_\cap$ satisfying all the planarity and the orthogonality constraints. 

For each exclusive edge $g \in E_1$, set $A(g) = A^\prime(g)$. Also, set $A(e) = A^\prime(e')$ and $A(f) = A^\prime(f')$. Finally, for each exclusive edge $g \in E_2 \cap E_2^\prime$, set $A(g) = A^\prime(g)$.

We prove that the planarity constraints for the edges of $G_1$ are satisfied by $A$.
Note that $e$ and $f$ do not alternate with each other since they are incident to the same vertex $v$. Also, if edge $e$ (edge $f$) alternates with an edge $g \in G_1$ along $G_\cap$, then edge $e'$ (edge $f'$) alternates with $g$ along $G_\cap^\prime$. Finally, any two edges different from $e$ and $f$ that alternate along $G_\cap$ also alternate along $G_\cap^\prime$. In all these cases, the planarity constraints are satisfied by $A$ since they are satisfied by $A^\prime$.

The planarity constraints for the edges of $G_2$ are satisfied by $A$ since any two of these edges alternate along $G_\cap$ if and only if they alternate along $G_\cap^\prime$, and since the planarity constraints are satisfied by $A^\prime$.

We finally prove that the orthogonality constraints are satisfied by $A$ at every vertex in $V$.
For the vertices in $V \setminus \{v\}$, this is true since they are satisfied by $A^\prime$ and since for every exclusive edge $g \in E_1 \cup E_2$ incident to these vertices, we have $g \in E_1^\prime \cup E_2^\prime$, by construction, and $A(g) = A^\prime(g)$. 
To prove that this is true also for $v$, we first argue that $A(e)=A^\prime(e'')$ and $A(f)=A^\prime(f'')$: By planarity constraints, we get $A'(x_1,v_a) = A'(v_a,x_5)=A'(x_6,u') = A'(u', x_{10})$, since they belong to a sequence of alternating edges; hence, by orthogonality constraints at $v_a$ and $u'$, we get $A'(e') = A'(x_1,v_a) = A'(v_a,x_5) = A'(x_6,u') = A'(u', x_{10})=A'(e'')$; since $A(e) = A'(e')$, by construction, we conclude $A(e) = A'(e'')$. The equality $A(f)=A^\prime(f'')$ follows symmetrically.
Hence, the orthogonality constraints at $v$ are satisfied by $A$ since they are satisfied at $v$ by $A^\prime$. This concludes the proof that \ourinstance{+}, which satisfies properties (i) and (iii), is equivalent to \ourinstance{}.
		
In order to construct an instance \ourinstance{*} equivalent to \ourinstance{+} that also satisfies property (ii), we observe that the proof of Lemma~\ref{le:reduction-degree-three-2-degree-three-outerplanar} can be easily extended so that it can be applied to \ourinstance{+}. This lemma, in fact, holds for instances \ourinstance{} satisfying property (i) and a property that is stronger than (iii), namely that $G_1$ has degree at most $3$. This stronger condition, however, is only used to ensure that there exists no exclusive edge in $E_1$ with an endpoint in $H_2$ and the other one not in $H_2$; refer to Fig.~\ref{fig:reduction-deg3-2-outerplanar-app}. In particular, it is used to ensure that there exists no edge connecting a vertex of $H_2$ to any of $u,v,w,z$. However, it is possible to prove that property (iii) is already sufficient to ensure the absence of these edges. Namely, suppose that there exists an edge in $E_1$ connecting a vertex $x$ of $H_2$ to vertex $v$, the other cases being analogous. This implies that $v$ has degree $4$ in $G_1$, since it is also adjacent to $u$. However, any path in cycle $G_\cap$ containing $u$, $x$, and $v$ also contains either $w$ or $z$, since $e$ and $f$ alternate along $G_\cap$; this is a contradiction to property (iii), since each of $w$ and $z$ is incident to an exclusive edge of $G_1$, namely $f$. This concludes the proof of the lemma.
\end{proof}
\end{document}